\documentclass[runningheads,a4paper]{llncs}

\usepackage{microtype} % Slightly tweak font spacing for aesthetics
\usepackage{amssymb,amsmath}
\usepackage{xspace}
\usepackage[usenames, dvipsnames]{color}
\usepackage{tcolorbox,xcolor}
\usepackage{graphicx}
\usepackage{hyperref}
\usepackage{framed}
\usepackage{algorithm, algorithmic}
\usepackage{subcaption}
\usepackage{wrapfig}
\usepackage{gensymb}
\usepackage{etoolbox}
\newcommand{\BeginMyEnumerate}{\begin{enumerate}\setlength{\itemsep}{-\parskip}}
\newcommand{\EndMyEnumerate}{\end{enumerate}}
\newcommand{\myenumerate}[1]{\BeginMyEnumerate #1 \EndMyEnumerate}

\newcommand{\mypara}[1]{\vspace{10pt} \noindent \textbf{\sffamily #1}}
\newcommand{\smallpara}[1]{\vspace{10pt} \noindent \emph{#1}}

\renewcommand{\le}{\leqslant}
\renewcommand{\ge}{\geqslant}
\renewcommand{\leq}{\leqslant}
\renewcommand{\geq}{\geqslant}
\newcommand{\eps}{\varepsilon}
\newcommand{\Reals}{{\Bbb R}}
\newcommand{\Nats}{{\Bbb N}}

\newcommand{\etal}{\emph{et al.\xspace}}

\newcommand{\Opt}{\mbox{\sc Opt}}
\newcommand{\LB}{\mbox{\sc lb}}

\newcommand{\runtime}{O}
\newcommand{\ds}{\mathcal{D}}
\newcommand{\tree}{\mathcal{T}}

\newcommand{\q}[1]{``#1''}

\newcommand{\pointset}{S}
\newcommand{\query}{Q}
\newcommand{\qrange}{Q}
\newcommand{\Sq}{\pointset_{\qrange}}
\newcommand{\qinterval}{\qrange:=[x,x']}

\newcommand{\setsize}{n}
\newcommand{\metric}{L}
\newcommand{\infmetric}{\metric_\infty}

\newcommand{\Lp}{\metric_p}

\newcommand{\point}{p}

\newcommand{\infdist}{d_\infty}
\newcommand{\maxlen}{L}

\newcommand{\corner}{c}

\newcommand{\region}{A}

\newcommand{\qpoint}{u}

\newcommand{\vor}{Vor}
\newcommand{\allclusterings}{\mathrm{Part}}
\newcommand{\costagg}{{\Phi}}   %  removed the subscript "agg" since we do not seem to need this

\newcommand{\clustering}{\mathcal{C}}
\newcommand{\C}{\clustering}        % shorthand
\newcommand{\clusteringopt}{\clustering_{\text{opt}}}
\newcommand{\Copt}{\clustering_{\text{opt}}}

\DeclareMathOperator*{\argmin}{arg\,min}
\newcommand{\ClusterQuery}{\mbox{{\sc ClusterQuery}}\xspace}
\newcommand{\CapacitatedClusterQuery}{\mbox{{\sc CapacitatedKCenterQuery}}\xspace}
\newcommand{\ssc}{\mbox{{\sc SingleShotClustering}}\xspace}
\newcommand{\pckc}{\mbox{{\sc 0/1-WeightedKCenter}}\xspace}

\newcommand{\DeltaSample}{\mbox{{\sc DeltaApprox}}\xspace}
\newcommand{\Decider}{\mbox{{\sc Decider}}\xspace}
\newcommand{\B}{\mathcal{B}}

\newcommand{\Z}{\mathcal{Z}}
\newcommand{\ranges}{\Z}
\newcommand{\range}{\sigma}

\newcommand{\diam}{\mathrm{diam}}
\newcommand{\radius}{\mathrm{radius}}
\newcommand{\size}{\mathrm{size}}
\newcommand{\Binner}{\B_\text{inner}}
\newcommand{\Bleaf}{\B_\text{leaf}}

\newcommand{\sample}{A}
\newcommand{\rpacking}{R}
\newcommand{\Aq}{\sample_{\qrange}}
\newcommand{\Rq}{\rpacking_{\qrange}}
\newcommand{\deltaq}{\delta_{\qrange}}
\newcommand{\polylog}{{\rm polylog}}
\newcommand{\per}{\mathrm{per}}
\newcommand{\splitset}{\mbox{\rm Split}(\qrange)}

\DeclareMathOperator*{\argmax}{{arg\,max}\xspace}
\definecolor{mypink}{rgb}{0.858, 0.188, 0.478}

\makeatletter
\define@key{Hyp}{footnotecolor}{%
 \HyColor@HyperrefColor{#1}\@footnotecolor%
}
\patchcmd{\@footnotemark}{\hyper@linkstart{link}}{\hyper@linkstart{footnote}}{}{}
\makeatother

\hypersetup{
     colorlinks      = true,
     citecolor       = ForestGreen, %PineGreen,
     linkcolor       = OrangeRed,
}

\begin{document}

\title{Range-Clustering Queries\thanks{MA is partly supported by Mikkel Thorup's Advanced Grant from the Danish Council
for Independent Research under the Sapere Aude research career programme. MdB, KB, MM,
and AM are supported by NWO grants 024.002.003, 612.001.207, 022.005.025, and 612.001.118
respectively.}}
\titlerunning{Range-Clustering Queries}

\author{Mikkel Abrahamsen\inst{1}, Mark de Berg\inst{2}, Kevin Buchin\inst{2}, Mehran Mehr\inst{2}, \and Ali~D.~Mehrabi\inst{2}}
\institute{Computer Science Department, University of Copenhagen, Denmark
           \and Computer Science Department, TU Eindhoven, the Netherlands}

\maketitle
%-----------------------------------------------------------------------------------
\begin{abstract}
In a geometric $k$-clustering problem the goal is to partition a set of points in
$\Reals^d$ into $k$ subsets such that a certain cost function of the clustering
is minimized.  We present data structures for orthogonal \emph{range-clustering queries}
on a point set $\pointset$: given a query box~$\qrange$ and an integer~$k\geq 2$,
compute an optimal $k$-clustering for $\pointset\cap Q$.
We obtain the following results.
\begin{itemize}
\item We present a general method to compute a $(1+\eps)$-approximation to a range-clustering
              query, where $\eps>0$ is a parameter that can be specified as part of the query. Our method
              applies to a large class of clustering problems, including $k$-center clustering in any $\Lp$-metric
              and a variant of $k$-center clustering where the goal is to minimize the sum (instead of maximum)
             of the cluster sizes.
\item We extend our method to deal with capacitated $k$-clustering problems, where each of the
              clusters should not contain more than a given number of points.
\item For the special cases of rectilinear $k$-center clustering in $\Reals^1$, and in $\Reals^2$
              for $k=2$ or~3, we present data structures that answer range-clustering queries exactly.
\end{itemize}
\end{abstract}
%-----------------------------------------------------------------------------------

%-----------------------------------------------------------------------------------
\section{Introduction}\label{sec:intro}
%-----------------------------------------------------------------------------------
The range-searching problem is one of the most important and widely studied problems
in computational geometry. In the standard setting one is given a set $\pointset$ of
points in $\Reals^d$, and a query asks to report or count all points inside
a geometric query range~$\qrange$. In many applications, however, one would like to perform
further analysis on the set $\pointset \cap \qrange$, to obtain more
information about its structure. Currently one then has to
proceed as follows: first perform a range-reporting query to explicitly
report~$\pointset \cap \qrange$, then apply a suitable analysis algorithm
to $\pointset \cap \qrange$. This two-stage process can be quite costly,
because algorithms for analyzing geometric data sets can be slow
and $\pointset\cap\qrange$ can be large. To avoid this we would need
data structures for what we call \emph{range-analysis queries}, which
directly compute the desired structural information about~$\pointset\cap\qrange$.
In this paper we develop such data structures for the case where one is
interested in a cluster-analysis of~$\pointset\cap\qrange$.

Clustering is a fundamental task in data analysis. It involves partitioning
a given data set into subsets called \emph{clusters}, such that
similar elements end up in the same cluster. Often the data
elements can be viewed as points in a geometric space, and similarity
is measured by considering the distance between the points.
We focus on clustering problems of the following type. Let $\pointset$
be a set of $n$~points in $\Reals^d$, and let $k\geq 2$ be a natural number.
A \emph{$k$-clustering} of $\pointset$ is a partitioning $\C$
of $\pointset$ into at most $k$ clusters. Let $\costagg(\C)$
denote the \emph{cost} of~$\C$. The goal is now to find
a clustering $\clustering$ that minimizes~$\costagg(\C)$. Many
well-known geometric clustering problems are of this type.
Among them is the $k$-center problem.  % and $k$-median, and $k$-means problems.
In the \emph{Euclidean $k$-center problem} $\costagg(\C)$
is the maximum cost of any of the clusters $C\in\C$, where the cost of~$C$
is the radius of its smallest enclosing ball.
Hence, in the Euclidean $k$-center problem we want to cover
the point set $\pointset$ by $k$ congruent balls of minimum radius.
The \emph{rectilinear $k$-center problem} is defined similarly except that
one considers the $\infmetric$-metric; thus we want to cover~$\pointset$
by $k$ congruent axis-aligned cubes\footnote{Throughout the paper,
when we speak of cubes (or squares, or rectangles, or boxes) we always mean axis-aligned cubes
(or squares, or rectangles, or boxes).}
of minimum size.
The $k$-center problem, including the important special case of the
2-center problem, has been studied extensively, both for the Euclidean case~(e.g.~\cite{abs-2cptd-13,c-mptwa-99,e-fcptc-97,hcl-gss-93,h-slt-05,s-nltp2c-97})
and for the rectilinear case~(e.g.~\cite{c-garot-99,sw-rppp-96}).
%
%The best known algorithm for the 2-center problem runs in $O(n\log^2 n)$ randomized time~\cite{e-fcptc-97},
%while for the general $k$-center problem the best algorithm runs in $O(n^{O(\sqrt{k})})$ time~\cite{hcl-gss-93}.
%There are also several approximation algorithms (e.g.~\cite{ap-eaac-02})
%

All papers mentioned above---in fact, all papers on clustering that we
know of---consider clustering in the single-shot version.
We are the first to study \emph{range-clustering queries} on a point set~$\pointset$: given
a query range~$\qrange$ and a parameter~$k$, solve the given
$k$-clustering problem on $\pointset\cap\qrange$.  We study this
problem for the case where the query range is an axis-aligned box.

%-----------------------------------------------------------------------------------
\mypara{Background.}
%-----------------------------------------------------------------------------------
Range-analysis queries can be seen as a very general form
of range-aggregate queries. In a range-aggregate query, the goal is to
compute some aggregate function $F(\pointset\cap\qrange)$ over the points
in the query range. The current state of the art typically deals with
simple aggregate functions of the following form:  each point $p\in \pointset$
has a \emph{weight}~$w(p)\in\Reals$, and
$F(\pointset\cap \qrange) := \bigoplus_{p\in \pointset\cap \qrange} w(p)$,
where~$\oplus$ is a semi-group operation.
Such aggregate functions are \emph{decomposable}, meaning that
$F(A\cap B)$ can be computed from $F(A)$ and $F(B)$, which makes them
easy to handle using existing data structures such as range trees.
%
%The difficulty with non-decomposable
%searching problems can be illustrated from the \emph{orthogonal colored range-counting
%problem}, where $D$ is set of colored points and one wants to count the number
%of distinct colors appearing in a query rectangle~$Q$.
%Here the best known solutions with polylogarithmic query time use at least quadratic
%storage~\citetwo{gjs-aghrs-96}{krsv-eocrc-08}, while for decomposable problems
%one easily gets polylogarithmic query times using only near-linear storage.

Only some, mostly recent, papers describe data structures supporting
non-decomposable analysis tasks. Several deal with finding the closest pair inside a query
range (e.g.~\cite{acfs-pspd-13,dgks-rl1m-14,gjks-dsrae-14}).
However, the closest pair does not give information about the global shape or distribution
of $\pointset\cap \qrange$, which is what our queries are about. The recent works by
Brass~\etal~\cite{bkssv-cats-13} and by Arya~\etal~\cite{amp-socg-15}
are more related to our paper.
Brass~\etal~\cite{bkssv-cats-13} present data structures for finding extent measures,
such the width, area or perimeter of the convex hull of $\pointset\cap\qrange$,
or the smallest enclosing disk.
(Khare~\etal~\cite{kams-ibsed-14} improve the result on smallest-enclosing-disk queries.)
These measures are strictly speaking not decomposable, but they depend only on the convex hull
of $\pointset\cap\qrange$ and convex hulls are decomposable.
A related result is by Nekrich and Smid~\cite{ns-coreset-cccg10},
who present a data structure that returns an $\eps$-coreset inside a query range.
%Note that the convex hull of $\pointset\cap\qrange$ or an $\eps$-coreset
%do not provide sufficient information to cluster $\pointset\cap\qrange$,
%making our range-clustering queries much more difficult to handle.
The measure studied by Arya~\etal~\cite{amp-socg-15}, namely the length
of the minimum spanning tree of $\pointset\cap\qrange$, cannot be computed form the convex
hull either: like our range-clustering queries, it requires more information
about the structure of the point set. Thus our paper continues
the direction set out by Arya~\etal, which is to design data structures
for more complicated analysis tasks on~$\pointset\cap\qrange$.
%
%Arya~\etal~\cite{amp-socg-15} present a data structure that can approximate
%the length of the MST of $\pointset\cap\qrange$; it uses an approach based on BBD-trees
%to obtain a structure with $O(n)$ storage and $O(\log n + (1/\eps_q \eps_w)^c)$
%query time, where $\eps_q$ and $\eps_w$ specify the approximation quality
%and $c$ is a constant depending on the dimension.

%-----------------------------------------------------------------------------------
\mypara{Contributions.}
%-----------------------------------------------------------------------------------
Our main result is a general method to answer \emph{approximate} orthogonal range-clustering
queries in $\Reals^d$. Here the query specifies (besides the query box~$\qrange$ and the number of
clusters $k$) a value $\eps>0$; the goal then is to compute a $k$-clustering~$\C$ of
$\pointset\cap\qrange$ with $\costagg(\C)\leq (1+\eps)\cdot \costagg(\Copt)$, where $\Copt$ is an
optimal clustering for~$\pointset\cap\qrange$. Our method works by computing a
sample~$R\subseteq\pointset\cap\qrange$
such that solving the problem on $R$ gives us the desired approximate solution.
We show that for a large class of cost functions $\costagg$ we can find such a sample
of size only $O(k(f(k)/\eps)^d)$, where $f(k)$ is a function that only depends on the
number of clusters. This is similar to the approach taken by Har-Peled and Mazumdar~\cite{hm-cc-stoc04},
who solve the (single-shot) approximate $k$-means and $k$-median problem efficiently
by generating a coreset of size $O((k/\eps^d)\cdot \log n)$.
A key step in our method is a procedure to efficiently compute a lower bound on
the value of an optimal solution within the query range. The class of clustering problems to which our method
applies includes the $k$-center problem in any $\Lp$-metric, variants of the
$k$-center problem where we want to minimize the sum (rather than maximum) of the
cluster radii, and the 2-dimensional problem where we want to minimize the
maximum or sum of the perimeters of the clusters. Our technique allows us,
for instance, to answer rectilinear $k$-center queries in the plane
in $O((1/\eps)\log n + 1/\eps^2)$ for $k=2$ or~3,
in $O((1/\eps)\log n + (1/\eps^2)\polylog(1/\eps))$ for $k=4$ or~5, and in
$O( (k/\eps) \log n  +  (k/\eps)^{O(\sqrt{k})})$ time for $k>3$.
We also show that for the rectilinear (or Euclidean) $k$-center problem,
our method can be extended to deal with the capacitated version of the problem.
In the capacitated version each cluster should not contain more
than~$\alpha\cdot(|\pointset\cap\qrange|/k)$ points, for a given ~$\alpha>1$.

In the second part of the paper we turn our attention to exact solutions to
range-clustering queries. Here we focus on rectilinear
\emph{$k$-center queries}---that is, range-clustering queries for the
rectilinear $k$-center problem---in $\Reals^1$ and~$\Reals^2$. We present
two linear-size data structures for queries in~$\Reals^1$; one has $O(k^2 \log^2 n)$
query time, the other has $O(3^k \log n)$ query time. For queries in $\Reals^2$
we present a data structure that answers 2-center queries
in $O(\log n)$ time, and one that answers 3-center queries in~$O(\log^2 n)$ time.
Both data structures use $O(n\log^{\eps} n)$ storage, where $\eps>0$ is an
arbitrary small (but fixed) constant.

%-----------------------------------------------------------------------------------
\section{Approximate Range-Clustering Queries}\label{sec:approx}
%-----------------------------------------------------------------------------------
In this section we present a general method to answer approximate range-clustering queries.
We start by defining the class of clustering problems to which it applies.
\medskip

Let $\pointset$ be a set of $n$ points in $\Reals^d$ and let $\allclusterings(\pointset)$
be the set of all partitions of $\pointset$. Let $\allclusterings_k(\pointset)$ be the
set of all partitions into at most $k$ subsets, that is, all $k$-clusterings of $\pointset$.
Let $\costagg:\allclusterings(S)\mapsto\Reals_{\ge0}$
be the cost function defining our clustering problem, and define
\[
\Opt_k(\pointset) := \min_{\C\in \allclusterings_k(S)} \costagg(\C)
\]
to be the minimum cost of any $k$-clustering. Thus the goal of a range-clustering
query with query range~$\query$ and parameter $k\geq 2$ is to compute
a clustering $\C\in\allclusterings_k(\Sq)$ such that
$\costagg(\C) = \Opt_k(\Sq)$, where $\Sq := \pointset\cap\query$. The method presented in this
section gives an approximate answer to such a query: for a given constant~$\eps>0$,
which can be specified as part of the query, the method will report a clustering
$\C\in\allclusterings_k(\Sq)$ with $\costagg(\C) \leq (1+\eps) \cdot \Opt_k(\Sq)$.

To define the class of clusterings to which our method applies, we will need the
concept of $r$-packings~\cite{h-11}. Actually, we will use a slightly weaker variant,
which we define as follows. Let $|pq|$ denote the Euclidean distance between two
points~$p$ and~$q$.
A subset $R\subseteq P$ of a point set $P$ is called a \emph{weak $r$-packing} for $P$,
for some $r>0$, if for any point $p\in P$ there exists a \emph{packing point} $q\in R$
such that $|pq|\leq r$. (The difference with standard $r$-packing is that we do not
require that $|qq'|>r$ for any two points $q,q'\in R$.)
%    \mdb{Check literature to see if the name weak $r$-packing has already been used.}
%    \mm{I found this paper. \url{http://www.mathunion.org/ICM/ICM1990.1/Main/icm1990.1.0511.0520.ocr.pdf}
%I think their definition is equivalent to ours.}
% In other words, $R$ is a weak $r$-packing for $\pointset$ if all the points of $\pointset$
% are within (Euclidean) distance $r$ of at least one packing point.
The clustering problems to which our method applies are the ones whose cost
function is \emph{regular}, as defined next.
%----------------------------------------------------------------------------------------------
\begin{definition}\label{def:cost-function}
A cost function $\costagg:\allclusterings(\pointset)\mapsto\Reals_{\ge0}$
is called \emph{$(c,f(k))$-regular}, if there is constant $c$ and
function~$f:\Nats_{\geq 2}\mapsto\Reals_{\geq 0}$ such that the following holds.
\begin{itemize}
  \item For any clustering $\C\in\allclusterings(\pointset)$, we have
       \[
       \costagg(\C)\geq c\cdot\max_{C\in\C} \diam(C),
       \]
       where $\diam(C) = \max_{p,q\in C} |pq|$ denotes the Euclidean diameter of the cluster~$C$.
       We call this the \emph{diameter-sensitivity} property.
  \item For any subset $\pointset'\subseteq \pointset$, any weak $r$-packing $R$ of $\pointset'$,
        and any $k\geq 2$, we have that
        \[
        \Opt_k(R)\leq\Opt_k(\pointset')\leq\Opt_k(R)+r\cdot f(k).
        \]
        Moreover, given a $k$-clustering $\C\in\allclusterings_k(R)$ we can compute a $k$-clustering
        $\C^*\in\allclusterings_k(\pointset')$ with $\costagg(\C^*) \leq \costagg(\C)+r\cdot f(k)$ in time $T_{\text{expand}}(n,k)$.
        We call this the \emph{expansion} property.
\end{itemize}
\end{definition}
%----------------------------------------------------------------------------------------------
\mypara{Examples.}
%----------------------------------------------------------------------------------------------
Many clustering problems have regular cost functions, in particular when the cost function
is the aggregation---the sum, for instance, or the maximum---of the costs
of the individual clusters. Next we give some examples.

\smallpara{The $k$-center problem in any $\Lp$-metric.}
For a cluster~$C$, let $\radius_p(C)$ denote the radius of the minimum
enclosing ball of $C$ in the $\Lp$-metric. In the $\infmetric$,
for instance, $\radius_p(C)$ is half the edge length of a minimum enclosing
axis-aligned cube of $C$. Then the cost of a clustering~$\C$
for the $k$-center problem in the $\Lp$-metric is
$\costagg^{\max}_p(\C) =\max_{C\in \C} \radius_p(C)$.
One easily verifies that the cost function for the rectilinear $k$-center problem
is $(1/(2\sqrt{d}),1)$-regular, and for the Euclidean $k$-center problem
it is $(1/2,1)$-regular. Moreover, $T_{\text{expand}}(n,k) = O(k)$
for the $k$-center problem, since we just have so scale each ball
by adding $r$ to its radius.\footnote{This time bound
only accounts for reporting the set of cubes that define the clustering.
If we want to report the clusters explicitly, we need to add an $O(n)$ term.}
(In fact $\costagg^{\max}_p(\C)$ is regular for any~$p$.)
% \mm{$(1/(2d),1)$ works, $(\sqrt[p]{d}/(2d),1)$ is tighter though!}.

\smallpara{Min-sum variants of the $k$-center problem.}
In the $k$-center problem the goal is to minimize $\max_{C\in \C} \radius_p(C)$.
Instead we can also minimize $\costagg^{\text{sum}}_p(\C) := \sum_{C\in \C}\radius_p(C)$, the
sum of the cluster radii. Also these costs functions are regular; the only difference is
that the expansion property is now satisfied with $f(k)=k$, instead of with $f(k)=1$.
Another interesting variant is to minimize $\left( \sum_{C\in \C}\radius_2(C)^2 \right)^{1/2}$,
which is $(1/(2\sqrt{d}),\sqrt{k})$-regular.
% $f(k)=\sqrt[p]{k}$ for the $L_p$ norm.

\smallpara{Minimum perimeter $k$-clustering problems.}
For a cluster $C$ of points in $\Reals^2$, define $\per(C)$ to be the
length of the perimeter of the convex hull of~$C$.
In the minimum perimeter-sum clustering problem the goal is to compute a
$k$-clustering $\C$ such that
$\costagg_\text{per} := \sum_{C\in\C}\per(C)$  is minimized~\cite{crw-91}.
This cost function is $(2,2\pi k)$-regular.
Indeed, if we expand the polygons in a clustering~$\C$ of a weak $r$-packing $R$ by taking the Minkowski
sum with a disk of radius~$r$, then the resulting shapes cover all the points
in $\pointset$. Each perimeter increases by $2\pi r$ in this process.
To obtain a clustering, we then assign each point to the cluster of its closest
packing point, so $T_{\text{expand}}(n,k) = O(n\log n)$.
%   \mdb{Probably we can somehow get rid of the $\log n$, but it's perhaps not really worth it.}

\smallpara{Non-regular costs functions.}
Even though many clustering problems have regular costs functions, not all clustering problems do.
For example, the $k$-means problem does not have a regular cost function. Minimizing the
the max or sum of the areas of the convex hulls of the clusters is not regular either.
\medskip

%----------------------------------------------------------------------------------------------
\mypara{Our data structure and query algorithm.}
%----------------------------------------------------------------------------------------------
We start with a high-level overview of our approach. Let $\pointset$ be the given point set on
which we want to answer range-clustering queries, and let $\query$ be the query range.
From now on we use $\Sq$ as a shorthand for~$\pointset\cap\query$.
We assume we have an algorithm $\ssc(P,k)$ that computes an optimal solution to the $k$-clustering
problem (for the given cost function~$\costagg$) on a given point set~$P$.
% In other words, $\ssc(P,k)$ solves the problem in the standard (non-query) setting.
(Actually, it is good enough if $\ssc(P,k)$ gives a $(1+\eps)$-approximation.)
Our query algorithm proceeds as follows.

%-----------------------------------------------------------------------------
\begin{algorithm}
\caption{$\ClusterQuery(k,\qrange,\eps)$.}
\myenumerate{
\item \label{step:lb} Compute a lower bound $\LB$ on $\Opt_k(\Sq)$.
\item \label{step:net} Set $r := \eps \cdot \LB /{f(k)}$ and compute a weak $r$-packing $R$ on $\Sq$.
\item \label{step:slow} $\C := \ssc(R,k)$.
\item \label{step:expand} Expand $\C$ into a $k$-clustering $\C^*$ of cost at most $\costagg(\C)+r\cdot f(k)$ for $\Sq$.
\item \label{step:return} Return $\C^*$.
}
\end{algorithm}
%-----------------------------------------------------------------------------

Note that Step~\ref{step:expand} is possible because $\costagg$ is $(c,f(k))$-regular.
The following lemma is immediate.

%--------------------------------------------------------------------------------------------
\begin{lemma}\label{le:global-alg}
$\costagg(\C^*) \leq (1+\eps) \cdot \Opt_k(\Sq)$.
\end{lemma}
%--------------------------------------------------------------------------------------------
%\begin{proof}
%We have
%\[
%\costagg(\C)+r\cdot f(k) = \Opt_k(R)+r\cdot f(k) \leq \Opt_k(\Sq) + \eps\cdot \LB \leq (1+\eps)\cdot\Opt_k(\Sq).
%\]
%\end{proof}
%--------------------------------------------------------------------------------------------

Next we show how to perform
Step~\ref{step:lb} and~\ref{step:net}: we will describe a data structure that allows us
to compute a suitable lower bound~$\LB$ and a corresponding weak $r$-packing, such that the size of
the $r$-packing depends only on $\eps$ and $k$ but not on $|\Sq|$.
% Step~\ref{step:slow}~and~\ref{step:expand} are problem-specific.
\medskip

Our lower bound and $r$-packing computations are based on so-called cube covers.
A \emph{cube cover} of~$\Sq$
is a collection $\B$ of interior-disjoint cubes
that together cover all the points in $\Sq$ and such that each $B\in\B$ contains at least
one point from ~$\Sq$ (in its interior or on its boundary).
Define the size of a cube~$B$, denoted by $\size(B)$, to be its edge length.
The following lemma follows immediately from the fact that the diameter of a cube~$B$ in $\Reals^d$
is $\sqrt{d}\cdot \size(B)$.

%--------------------------------------------------------------------------------------------
\begin{lemma}\label{lem:rnet}
Let $\B$ be a cube cover of $\Sq$ such that $\size(B)\leq r/\sqrt{d}$
for all $B\in\B$. Then any subset $R\subseteq \Sq$ containing a point from each cube $B\in\B$
is a weak $r$-packing for $\pointset$.
\end{lemma}
%--------------------------------------------------------------------------------------------
Our next lemma shows we can find a lower bound on
$\Opt_k(\Sq)$ from a suitable cube cover.

%--------------------------------------------------------------------------------------------
\begin{lemma}\label{lem:optbound}
Suppose the cost function $\costagg$ is $(c,f(k))$-regular.
Let $\B$ be a cube cover of $\Sq$ such that $|\B|> k 2^d$.
Then $\Opt_k(\Sq) \geq c\cdot\min_{B\in\B}\size(B)$.
\end{lemma}
%--------------------------------------------------------------------------------------------
\begin{proof}
For two cubes $B$ and $B'$ such that the maximum $x_i$-coordinate of $B$ is at most the minimum
$x_i$-coordinate of~$B'$, we say that $B$ is \emph{$i$-below} $B'$ and $B'$ is \emph{$i$-above}~$B$.
We denote this relation by $B \prec_i B'$. Now consider an optimal $k$-clustering $\clusteringopt$
of~$\Sq$. By the pigeonhole principle, there is a cluster $C\in \clusteringopt$ containing
points from at least $2^d+1$ cubes. Let $\B_C$ be the set of cubes that contain at least one point in~$C$.

Clearly, if there are cubes $B,B',B''\in\B_C$ such that $B'\prec_i B\prec_i B''$
for some $1\leq i\leq d$, then the cluster~$C$ contains two points (namely from $B'$ and $B''$)
at distance at least $\size(B)$ from each other. Since $\costagg$ is $(c,f(k))$-regular this
implies that $\costagg(\clusteringopt)\geq c\cdot\size(B)$, which proves the lemma.

Now suppose for a contradiction that such a triple $B',B,B''$ does not exist.
Then we can define a characteristic vector $\Gamma(B) = (\Gamma_1(B),\ldots,\Gamma_d(B))$
for each cube $B\in\B_C$, as follows:
\[
\Gamma_i(B) = \left\{ \begin{array}{ll}
                    0 & \mbox{if no cube $B'\in\B_C$ is $i$-below $B$} \\
                    1 & \mbox{otherwise}
                    \end{array}
             \right.
\]

Since the number of distinct characteristic vectors is $2^d < |B_C|$, there must be
two cubes $B_1,B_2\in B_C$ with identical characteristic vectors. However, any two
interior-disjoint cubes can be separated by an axis-aligned hyperplane,
so there is at least one $i\in\{1,\ldots,d\}$ such that
$B_1\prec_i B_2$ or $B_2\prec_i B_1$.
Assume without loss of generality that $B_1\prec_i B_2$, so $\Gamma_i(B_2)=1$.
Since $\Gamma(B_1)=\Gamma(B_2)$ there must be a cube $B_3$ with $B_3\prec_i B_1$.
But then we have a triple $B_3\prec_i B_1\prec_i B_2$, which is a contradiction.
\end{proof}
%--------------------------------------------------------------------------------------------

Next we show how to efficiently perform Steps~\ref{step:lb} and~\ref{step:net}
of~$\ClusterQuery$. Our algorithm uses a compressed octree $\tree(\pointset)$
on the point set~$\pointset$, which we briefly describe next.

For an integer $s$, let $G_s$ denote the grid in $\Reals^d$ whose cells have size
$2^{s}$ and for which the origin $O$ is a grid point. A \emph{canonical cube} is
any cube that is a cell of a grid $G_s$, for some integer $s$.
A \emph{compressed octree} on a point set~$\pointset$ in $\Reals^d$ contained in
a canonical cube~$B$ is a tree-like structure defined recursively, as follows.
\begin{itemize}
\item If $|\pointset|\leq 1$, then $\tree(\pointset)$ consists of a single leaf node, which
corresponds to the cube~$B$.
\item If $|\pointset|> 1$, then consider the cubes $B_1,\ldots,B_{2^d}$ that result from
      cutting $B$ into $2^d$ equal-sized cubes.
    \begin{itemize}
    \item If at least two of the cubes~$U_i$ contain at least one point
        from~$\pointset$ then $\tree(\pointset)$
        consists of a root node with $2^d$ children $v_1,\ldots,v_{2^d}$,
        where $v_i$ is the root of a compressed octree for\footnote{Here we
        assume that points on the boundary between cubes
        are assigned to one of these cubes in a consistent manner.} $B_i\cap \pointset$.
    \item If all points from~$\pointset$ lie in the same cube~$B_i$, then let
        $B_\text{in}\subseteq B_i$ be the smallest canonical cube containing
        all points in~$\pointset$. Now $\tree(\pointset)$ consists of a root node with
        two children: one child~$v$ which is the root of a compressed octree for $\pointset$
        inside $B_\text{in}$, and one leaf node~$w$ which represents the
        donut region $B\setminus B_\text{in}$.
    \end{itemize}
\end{itemize}

A compressed octree for a set $\pointset$ of $n$ points can be computed in $O(n\log n)$
time, assuming a model of computation where the smallest canonical cube of two points
can be computed in~$O(1)$ time~\cite[Theorem~{2.23}]{h-11}.
For a node $v\in \tree(\pointset)$, we denote the cube or donut corresponding to~$v$
by $B_v$, and we define $\pointset_v := B_v \cap \pointset$.
It will be convenient to slightly modify the compressed quadtree by removing all
nodes~$v$ such that $\pointset_v=\emptyset$. (These nodes must be leaves.) Note that
this removes all nodes~$v$ such that $B_v$ is a donut. As a result, the parent of
such a donut node now has only one child,~$w$; we remove~$w$ and
link the parent of $w$ directly to $w$'s (non-empty) children. The modified
tree~$\tree(\pointset)$---with a slight abuse of terminology we still refer to
$\tree(\pointset)$ as a compressed octree---has the property that any internal
node has at least two children. We augment $\tree(\pointset)$ by
storing at each node~$v$ an arbitrary point~$p\in B_v\cap \pointset$.
\medskip

Our algorithm descends into~$\tree(\pointset)$ to
find a cube cover $\B$ of $\Sq$ consisting of canonical cubes,
such that $\B$ gives us a lower bound on $\Opt_k(\Sq)$.
In a second phase, the algorithm then refines the cubes in the cover
until they are small enough so that, if we select one point from each cube,
we get a weak $r$-packing of $\Sq$ for the appropriate value of~$r$.
The details are described in Algorithm~\ref{alg:net}, where we assume
for simplicity that $|\Sq|>1$.
(The case $|\Sq|\leq 1$ is easy to check and handle.
In addition, the algorithm will need several supporting data structures,
which we will describe them later.)

%-----------------------------------------------------------------------------
\begin{algorithm}
\caption{Algorithm for steps~\ref{step:lb} and~\ref{step:net} of~$\ClusterQuery$,
for a $(c,f(k))$-regular cost function.}
\label{alg:net}
\myenumerate{
\item $\Binner := B_{\text{root}(\tree(\pointset))}$ and $\Bleaf := \emptyset$.
\item $\rhd$ Phase~1: Compute a lower bound on $\Opt_k(\Sq)$.
\item {\bf While} {$|\Binner \cup \Bleaf| \leq k2^{2d}$ and $\Binner \neq \emptyset$} {\bf do} \label{line:while1}
        \myenumerate{
        \item[(i)] Remove a largest cube $B_v$ from $\Binner$. Let $v$ be the corresponding node.
        \item[(ii)] {\bf If} {$B_v \not\subseteq \qrange$} {\bf then} \label{line:split-start}
        \myenumerate{
            \item[(i)] Compute $\text{bb}(\Sq\cap B_v)$, the bounding box of $\Sq\cap B_v$.
            \item[(ii)] Find the deepest node $u$ such that $\text{bb}(\Sq\cap B_v) \subseteq B_{u}$ and set $v := u$.
            }
        \item[(iii)] {\bf EndIf}
        \item[(iv)] For each child $w$ of $v$ such that $B_w\cap\Sq \neq\emptyset$, insert $B_w$ into
        $\Binner$ if $w$ is an internal node and insert $B_w$ into $\Bleaf$ if $w$
        is a leaf node. \label{line:split-end}
        }
\item  {\bf EndWhile}
\item $\LB := c\cdot \max_{B_v\in \Binner} \size(B_v)$ \label{line:LB}.
\item $\rhd$ Phase~2: Compute a suitable weak $r$-packing.
\item $r := \eps\cdot \LB/f(k)$.
\item {\bf While} {$\Binner\neq \emptyset$} {\bf do} \label{line:loop2-start}
        \myenumerate{
            \item[(i)] Remove a cube $B_v$ from $\Binner$ and handle it as in
            lines~\ref{line:split-start}--\ref{line:split-end}, with
            the following change:
            if $\size(B_w) \leq r/\sqrt{d}$ then always insert $B_w$
            into $\Bleaf$ (not into $\Binner$).
            }
\item {\bf EndWhile} \label{line:loop2-end}
\item For each cube~$B_v\in \Bleaf$ pick a point in $\Sq\cap B_v$ and put it into~$\Rq$. \label{line:pick}
\item Return $\Rq$.
}
\end{algorithm}
%-----------------------------------------------------------------------------

Note that we continue the loop in lines~\ref{line:while1}--\ref{line:split-end}
until we collect $k 2^{2d}$ cubes (and not $k2^d$, as Lemma~\ref{lem:optbound}
would suggest) and that in line~\ref{line:LB} we take the maximum cube
size (instead of the minimum, as Lemma~\ref{lem:optbound} would suggest).

%--------------------------------------------------------------------------------------------
\begin{lemma}\label{le:correctness}
The value $\LB$ computed by Algorithm~\ref{alg:net} is a correct lower bound on
$\Opt_k(\Sq)$. In addition, the set $\Rq$ is a weak $r$-packing for $r = \eps\cdot\LB/f(k)$ of
size~$O(k (f(k)/(c\;\eps))^d)$.
\end{lemma}
%--------------------------------------------------------------------------------------------
\begin{proof}
As the first step to prove that $\LB$ is a correct lower bound, we claim that the
loop in lines~\ref{line:while1}--\ref{line:split-end} maintains the following invariant:
(i) $\bigcup (\Binner \cup \Bleaf)$ contains all points in $\Sq$, and
(ii) each $B\in\Binner$ contains at least two points from $\Sq$ and each
$B\in\Bleaf$ contains exactly one point from $\Sq$.
This is trivially true before the loop starts, under our assumption that $|\Sq|\geq 2$.
Now suppose we handle a cube $B_v\in\Binner$. If $B_v \subseteq \qrange$
then we insert the cubes $B_w$ of all children into $\Binner$ or $\Bleaf$,
which restores the invariant. If $B_v \not\subseteq \qrange$ then we first replace $v$ by~$u$.
The condition $\text{bb}(\Sq\cap B_v) \subseteq B_{u}$ guarantees that all points
of $\Sq$ in $B_v$ are also in $B_u$. Hence, if we then insert the cubes $B_w$ of $u$'s children
into $\Binner$ or $\Bleaf$, we restore the invariant. Thus at any time, and in particular
after the loop, the set $\Binner \cup \Bleaf$ is a cube cover of $\Sq$.

To complete the proof that $\LB$ is a correct lower bound we do not work with the set
$\Binner\cup\Bleaf$ directly, but  we work with a set $\B$ defined as follows.
For a cube $B_v\in\Binner\cup\Bleaf$, define $\text{parent}(B_v)$ to be the cube $B_u$
corresponding to the parent node~$u$ of~$v$. For each cube $B_v\in\Binner\cup\Bleaf$
we put one cube into $\B$, as follows. If there is another cube $B_w\in\Binner\cup\Bleaf$
such that $\text{parent}(B_w)\subsetneq\text{parent}(B_v)$, then we put $B_v$ itself into $\B$,
and otherwise we put $\text{parent}(B_v)$ into $\B$. Finally, we remove all duplicates from~$\B$.
Since $\Binner\cup\Bleaf$ is a cube cover for $\Sq$---that is, the cubes in
$\Binner\cup\Bleaf$ are disjoint and they cover all points in $\Sq$---the same is true for~$\B$.
Moreover, the only duplicates in $\B$ are cubes that are the parent of multiple nodes
in $\Binner\cup\Bleaf$, and so $|\B|\geq|\Binner\cup\Bleaf |/2^d > k2^d$.
By Lemma~\ref{lem:optbound} we have $\Opt_k(\Sq)\geq c\cdot\min_{B_v\in\B}\size(B_v)$.

It remains to argue that $\min_{B_v\in\B}\size(B_v)\geq\max_{B_v\in\Binner}\size(B_v)$.
We prove this by contradiction. Hence, we assume
$\min_{B_v\in\B}\size(B_v)<\max_{B_v\in\Binner}\size(B_v)$ and we
define $B:=\argmin_{B_v\in\B} \size(B_v)$ and $B':=\argmax_{B_v\in\Binner} \size(B_v)$.
Note that for any cube $B_v\in\B$ either $B_v$ itself is in $\Binner\cup\Bleaf$ or
$B_v=\text{parent}(B_w)$ for some cube $B_w\in\Binner\cup\Bleaf$.
We now make the following case distinction.

\textsc{Case I:} $B=\text{parent}(B_w)$ for some cube $B_w\in\Binner\cup\Bleaf$.
But this is an immediate contradiction since Algorithm~\ref{alg:net} would have to
split $B'$ before splitting $B$.
%(Note that $B_w$ is a witness that $B$ has been split by Algorithm~\ref{alg:net}.)}

\textsc{Case II:} $B\in\Binner\cup\Bleaf$.
Because $B$ itself was put into $\B$ and not $\text{parent}(B)$, there exists a cube
$B_w\in\Binner\cup\Bleaf$ such that $\text{parent}(B)\supsetneq\text{parent}(B_w)$,
which means $\text{size(parent(}B_w))<\text{size(parent(}B))$.
In order to complete the proof, it suffices to show that $\text{size(parent(}B_w))\leq\text{size}(B)$.
Indeed, since $B'$ has not been split by Algorithm~\ref{alg:net} (because $B'\in\Binner$)
we know that $\text{size}(B')\leq\text{size(parent(}B_w))$.
This inequality along with the inequality $\text{size(parent(}B_w))\leq\text{size}(B)$
imply that $\text{size}(B')\leq\text{size}(B)$ which is in contradiction with
$\text{size}(B)<\text{size}(B')$.
To show that $\text{size(parent(}B_w))\leq\text{size}(B)$ we consider the following
two subcases.
(i) $\text{parent}(B)$ is a degree-1 node.
This means that $\text{parent}(B)$ corresponds to a cube that was split into a
donut and the cube corresponding to $B$.
Since the cube corresponding to $B_w$ must be completely inside the cube corresponding
to $\text{parent}(B)$ (because $\text{size(parent(}B_w))<\text{size(parent(}B))$) and a
donut is empty we conclude that the cube corresponding to $B_w$ must be completely
inside the cube corresponding to $B$.
Hence, $\text{size(parent(}B_w))\leq\text{size}(B)$.
(ii) $\text{parent}(B)$ is not a degree-1 node.
The inequality $\text{size(parent(}B_w))<\text{size(parent(}B))$ along with the fact
that $\text{parent}(B)$ is not a degree-1 node imply that
$\text{size(parent(}B_w))\leq\text{size}(B)$.

This completes the proof that $\LB$ is a correct lower bound.
Next we prove that $\Rq$ is a weak $r$-packing for $r = \eps\cdot\LB/f(k)$.
Observe that after the loop in lines~\ref{line:loop2-start}--\ref{line:loop2-end},
the set $\Bleaf$ is still a cube cover of $\Sq$. Moreover, each cube~$B_v\in\Bleaf$
either contains a single point from $\Sq$ or its size is at most $r/\sqrt{d}$.
Lemma~\ref{lem:rnet} then implies that $\Rq$ is a weak $r$-packing for the desired
value of~$r$.
\medskip

It remains to bound the size of~$\Rq$. To this end we note that at each iteration
of the loop in lines~\ref{line:while1}--\ref{line:split-end} the size of $\Binner\cup\Bleaf$
increases by at most $2^d-1$, so after the loop we have
$|\Binner\cup\Bleaf|\leq k2^{2d}+2^d-1$. The loop in
lines~\ref{line:loop2-start}--\ref{line:loop2-end} replaces each cube
$B_v \in\Binner$ by a number of smaller cubes. Since
$\LB = c\cdot\max_{B_v\in\Binner} \size(B_v)$ and $r=\eps\cdot\LB/f(k)$,
each cube $B_v$ is replaced by only $O((f(k)2^d\sqrt{d}/(c\;\eps))^d)$ smaller cubes.
Since $d$ is a fixed constant, the total number of cubes we end up with
(which is the same as the size of the $r$-packing) is $O(k (f(k)/(c\;\eps))^d)$.
\end{proof}
%--------------------------------------------------------------------------------------------

Lemma~\ref{le:correctness}, together with Lemma~\ref{le:global-alg}, establishes the correctness
of our approach. To achieve a good running time, we need a few supporting data structures.
\begin{itemize}
\item We need a data structure that can answer the following queries: given a query box~$Z$,
      find the deepest node $u$ in $\tree(\pointset)$ such that $Z\subseteq B_u$. With a
      \emph{centroid-decomposition tree} $\tree_\text{cd}$
      we can answer such queries in $O(\log n)$ time.
A centroid-decomposition tree $\tree_\text{cd}$
on the compressed octree $\tree(\pointset)$ is defined as follows.
View $\tree(\pointset)$ as an acyclic graph of maximum degree~$2^d+1$.
Let $v^*$ be a centroid of $\tree(\pointset)$, that is, $v^*$ is a node
whose removal splits $\tree(\pointset)$ into at most $2^d+1$ subgraphs
each containing at most half the nodes.
We recursively construct centroid-decomposition trees
$\tree_\text{cd}^1,\tree_\text{cd}^2,\ldots $ for each of the subgraphs.
The centroid-decomposition tree~$\tree_\text{cd}$ now consists of a root node
corresponding to $v^*$ that has $\tree_\text{cd}^1,\tree_\text{cd}^2,\ldots$
as subtrees. Note that one of the subtrees corresponds to the region
outside~$B_{v^*}$, while the other subtrees correspond to regions inside
$B_{v^*}$ (namely the cubes of the children of $v^*$ in $\tree(\pointset)$).

With $\tree_\text{cd}$ we can answer the following queries in $O(\log n)$
time: given a query box~$Z$, find the deepest node $u$ in $\tree(\pointset)$ such
that $Z\subseteq B_u$. We briefly sketch the (standard) procedure for this.
First, check if $Z\subseteq B_{v^*}$, where $v^*$ is
the node of $\tree(\pointset)$ corresponding to the root of~$\tree_\text{cd}$.
If not, recursively search the subtree $\tree_\text{cd}^j$ corresponding to the
region outside~$B_{v^*}$. If $Z\subseteq B_{v^*}$ then check if $v^*$ has
a child~$w$ such that $Z\subseteq B_{w}$; if so, recurse on the corresponding
subtree $\tree_\text{cd}^{j'}$, and otherwise report $B_{v^*}$ as the answer.

\item We need a data structure $\ds$ that can answer the following queries on $\pointset$:
      given a query box~$Z$ and an integer $1\leq i\leq d$,
      report a point in $\pointset\cap Z$ with maximum $x_i$-coordinate, and one with minimum
      $x_i$-coordinate. It is possible to answer such queries in $O(\log^{d-1}n)$ time with
      a range tree (with fractional cascading), which uses $O(n\log^{d-1} n)$ storage.
      Note that this also allows us to compute
      the bounding box of $\pointset\cap Z$ in $O(\log^{d-1}n)$ time.
      (In fact slightly better bounds are possible~\cite{ls-updsrr-94}, but for simplicity
      we stick to using standard data structures.)
      % ls-updsrr-94 gives: O(log^{d-2}n loglogn + log n) query, O(n log^{d-2}n) storage
      % and O(n log^{d-2}n loglogn) expected preprocessing
\end{itemize}
%--------------------------------------------------------------------------------------------
\begin{lemma} \label{le:correctness2}
Algorithm~\ref{alg:net} runs in
$O(k\left(f(k)/(c\;\eps)\right)^d+k\left((f(k)/(c\;\eps))\log n\right)^{d-1})$ time.
\end{lemma}
%--------------------------------------------------------------------------------------------
\begin{proof}
We store the set $\Binner$ in a priority queue base on the size of the cubes,
so we can remove the cube of maximum size in $O(\log n)$ time. To handle
a cube~$B_v$ in an iteration of the first while loop we need $O(\log^{d-1}n)$ time,
which is the time needed to compute $\text{bb}(\Sq\cap B_v)$
using our supporting data structure~$\ds$. Next observe that each iteration of the
loop increases the size of $\Binner\cup\Bleaf$.
When $B_v \subseteq\query$ this is clear, since every internal node in $\tree(\pointset)$
has at least two children. When $B_v \not\subseteq\query$ we first replace $v$
by the deepest node $u$ such that $\text{bb}(\Sq\cap B_v)\subseteq B_{u}$.
This ensures that at least two of the children of $u$ must contain a point in $\Sq$,
so the size of $\Binner\cup\Bleaf$ also increases in this case.
We conclude that the number of iterations is bounded by~$k2^{2d}$, and so the running
time for Phase~1 is $O(k 2^{2d}\log^{d-1}n)$.

To bound the time for Phase~2 we observe that the computation of
$\text{bb}(\Sq\cap B_v)$ is only needed when $B_v \not\subseteq \query$.
Similarly, we only need our supporting data structure~$\ds$ for picking
a point from $\Sq\cap B_v$ in line~\ref{line:pick} when $B_v \not\subseteq\query$.
The total number of cubes that are handled and generated in Phase~2 is
$O(k (f(k)/(c\;\eps))^d)$, but the number of cubes that intersect the boundary
of the query range~$\query$ is a factor $f(k)/\eps$ smaller. Thus the
total time for Phase~2 is
$O(k\left(f(k)/(c\;\eps)\right)^d+k\left((f(k)/(c\;\eps))\right)^{d-1} \log^{d-1}n )$.
\end{proof}
%--------------------------------------------------------------------------------------------

This leads to the following theorem (where we use that $T_{\mathrm{ss}}$ is at least linear).
%--------------------------------------------------------------------------------------------
\begin{theorem}
\label{thm:noncapacitated}
Let $\pointset$ be a set of $n$ points in $\Reals^d$ and let $\costagg$ be a $(c,f(k))$-regular
cost function. Suppose we have an algorithm that solves the given
clustering problem on a set of $m$ points in $T_{\mathrm{ss}}(m)$ time.
Then there is a data structure that uses $O(n\log^{d-1}n)$ storage
such that, for a query range~$\query$ and query values $k\geq 2$
and $\eps>0$, we can compute a $(1+\eps)$-approximate answer to a range-clustering query in time
\[
O\left(
 %  k \left( \frac{f(k)}{c\;\eps}\right)^d +  OMITTED BECAUSE IT IS DOMINATED BY T_ss(...)
      k \left( \frac{f(k)}{c\;\eps} \cdot \log n \right)^{d-1}
   +  T_{\mathrm{ss}} \left( k \left( \frac{f(k)}{c\;\eps}\right)^d,k \right)
   +  T_{\mathrm{expand}}(n,k)
\right).
\]
\end{theorem}
%--------------------------------------------------------------------------------------------

As an example application we consider $k$-center queries in the plane.
(The result for rectilinear 2-center queries is actually inferior to the
exact solution presented later.)
%--------------------------------------------------------------------------------------------
\begin{corollary}
\label{cor:noncapacitated}
Let $\pointset$ be a set of $n$ points in $\Reals^2$. There is a data structure that uses
$O(n\log n)$ storage such that, for a query range~$\query$ and query values $k\geq 2$
and $\eps>0$, we can compute a $(1+\eps)$-approximate answer to a $k$-center query within the
following bounds:
\begin{enumerate}
\item[(i)] For the rectilinear case with $k=2$ or~3, the query time is $O((1/\eps)\log n + 1/\eps^2)$;
\item[(ii)] For the rectilinear case with $k=4$ or~5, the query time is
\[
O((1/\eps)\log n + (1/\eps^2)\cdot\polylog(1/\eps));
\]
\item[(iii)] For the Euclidean case with $k=2$, the expected query time is
\[
O((1/\eps)\log n + (1/\eps^2)\log^2(1/\eps));
\]
\item[(iv)]  For the rectilinear case with $k>5$ and the Euclidean case with $k>2$ the query time is
$O( (k/\eps) \log n  +  (k/\eps)^{O(\sqrt{k})})$.
\end{enumerate}
\end{corollary}
%--------------------------------------------------------------------------------------------
\begin{proof}
Recall that the cost function for the $k$-center problem is $(1/(2\sqrt{d}),1)$-regular for
the rectilinear case and $(1/2,1)$-regular for the Euclidean case. We now obtain
our results by plugging in the appropriate algorithms for the single-shot version.
For (i) we use the linear-time  algorithm of Hoffmann~\cite{h-slt-05},
for (ii) we use the $O(n\cdot\polylog n)$-time algorithms of Sharir and Welzl~\cite{sw-rppp-96},
  % for k=4 it is actually O(n log n)
for (iii) we use the $O(n\log^2 n)$-time randomized algorithm of Eppstein~\cite{e-fcptc-97},
and for (iv) we use the $n^{O(\sqrt{k})}$-time algorithm of Agarwal and Procopiuc~\cite{ap-algo-02}.
\end{proof}

%--------------------------------------------------------------------------------------------
\section{Approximate Capacitated \texorpdfstring{$k$}{k}-Center Queries}\label{sec:capacitated}
%--------------------------------------------------------------------------------------------
In this section we study the capacitated variant of the rectilinear $k$-center problem in the plane.
In this variant we want to cover a set $\pointset$ of $n$ points in $\Reals^2$ with
$k$~congruent squares of minimum size, under the condition that
no square is assigned more than $\alpha\cdot n/k$ points, where $\alpha>1$ is a given constant.
For a capacitated rectilinear $k$-center query this means
we want to assign no more than $\alpha\cdot |S_Q|/k$ points to each square.
Our data structure will report a $(1+\eps,1+\delta)$-approximate answer to capacitated
rectilinear $k$-center queries: given a query range~$\query$, a natural number~$k\geq 2$,
a constant~$\alpha>1$, and real numbers $\eps,\delta>0$, it computes a
set~$\C=\{b_1,\ldots,b_k\}$ of congruent squares such that:
\begin{itemize}
\item each $b_i$ can be associated to a subset $C_i\subseteq \Sq\cap b_i$ such that
      $\{C_1,\ldots,C_k\}$ is a $k$-clustering of $\Sq$ and
      $|C_i|\leq (1+\delta) \alpha\cdot |S_Q|/k$; and
\item the size of the squares in $\C$ is at most $(1+\eps)\cdot \Opt_k(\Sq,\alpha)$,
      where $\Opt_k(\Sq,\alpha)$ is the value of an optimal solution to the problem on $\Sq$
      with capacity upper bound~$U_Q := \alpha\cdot |S_Q|/k$.
\end{itemize}

Thus we allow ourselves to violate the capacity constraint by a factor~$1+\delta$.
\medskip

To handle the capacity constraints, it is not sufficient to work with $r$-packings---we
also need $\delta$-approximations. Let $P$ be a set of points in $\Reals^2$.
% and let $\Sigma$ be the family of all axis-aligned rectangles.
A \emph{$\delta$-approximation} of $P$ with respect to axis-aligned rectangles is a
subset $\sample\subseteq P$ such that for any rectangle $\range$ we have
\[
\big| \ |P\cap\range| / |P| \; - \; |\sample\cap\range| / |\sample| \big| \le \delta
\]

From now on, whenever we speak of $\delta$-approximations, we mean $\delta$-approximations
with respect to rectangles. Our method will use a special variant of the capacitated
$k$-center problem, where we also have points that must be covered but do not count
for the capacity:
%--------------------------------------------------------------------------------------------
\begin{definition}
Let $R\cup A$ be a point set in $\Reals^2$, $k\geq 2$ a natural number, and $U$ a capacity bound.
The \emph{0/1-weighted capacitated $k$-center problem} in $\Reals^2$ is to compute
a set~$\C=\{b_1,\ldots,b_k\}$ of congruent squares of minimum size where
each $b_i$ is associated to a subset $C_i\subseteq (R\cup A)\cap b_i$ such that
$\{C_1,\ldots,C_k\}$ is a $k$-clustering of $R\cup A$ and $|C_i\cap A|\leq U$.
\end{definition}
%--------------------------------------------------------------------------------------------
%Given a query range $\query$ and a natural number $k$ and a capacity $L\ge n/k$, we solve a partially
%capacitated rectilinear $k$-center problem on $\pointset'=R\cup \Aq$ and capcity $L'=L|\Aq|/|\pointset_\query|$,
%where $R$ is the points in a weak $r$-packing on $\pointset_\query$ and $\Aq$ is a $\deltaq$-approximation on $\pointset'$
%for $r=\eps \cdot LB/2$ and $\deltaq=\delta/4k^3$. We show that after expansion of the squares of the solution
%to this partially capacitated problem by $r$ from each side, we have an $(1+\eps,1+\delta)$-approximate solution for the original problem.

For a square~$b$, let $\text{expand}(b,r)$ denote the square $b$ expanded by~$r$ on each side
(so its radius in the $\infmetric$-metric increases by~$r$). Let
$\pckc$ be an algorithm for the single-shot capacitated rectilinear $k$-center problem.
Our query algorithm is as follows.

%-----------------------------------------------------------------------------
\begin{algorithm}
\caption{$\CapacitatedClusterQuery(k,\query,\alpha,\eps,\delta)$.}
\myenumerate{
\item \label{step:cap-lb} Compute a lower bound $\LB$ on $\Opt_k(\Sq)$.
\item \label{step:cap-net} Set $r := \eps \cdot \LB /{f(k)}$ and compute a weak $r$-packing $R$ on $\Sq$.
\item \label{step:cap-approx} Set $\deltaq := \delta/{16k^3}$ and compute a $\deltaq$-approximation $\Aq$ on $\Sq$.
\item \label{step:cap-slow} Set $U := (1+\delta/2)\cdot\alpha\cdot|\Aq|/k$ and $\C := \pckc(R\cup\Aq,k,U)$.
\item \label{step:cap-expand} $\C^* := \{ \text{expand}(b,r) : b \in \C \}$.
\item \label{step:cap-return} Return $\C^*$.
}
\end{algorithm}
%-----------------------------------------------------------------------------

%--------------------------------------------------------------------------------------------
%\begin{quotation}
%\noindent
%$\CapacitatedClusterQuery(k,\query,\alpha,\eps,\delta)$ \\[-5mm]
%\begin{algorithmic}[1]
%  \State \label{step:cap-lb} Compute a lower bound $\LB$ on $\Opt_k(\Sq)$.
%  \State \label{step:cap-net} Set $r := \eps \cdot \LB /{f(k)}$ and compute a weak $r$-packing $R$ on $\Sq$.
%  \State \label{step:cap-approx} Set $\deltaq := \delta/{16k^3}$ and compute a $\deltaq$-approximation $\Aq$ on $\Sq$.
%  \State  \label{step:cap-slow} Set $U := (1+\delta/2)\cdot\alpha\cdot|\Aq|/k$ and
%          $\C := \pckc(R\cup\Aq,k,U)$.
%  \State \label{step:cap-expand} $\C^* := \{ \text{expand}(b,r) : b \in \C \}$.
%  \State \label{step:cap-return} \Return $\C^*$.
%\end{algorithmic}
%\end{quotation}
%%\caption{caption}
%%\end{algorithm}
%--------------------------------------------------------------------------------------------

Note that the lower bound computed in Step~\ref{step:cap-lb} is a lower bound on
the uncapacitated problem (which is also a lower bound for the capacitated problem).
Hence, for Step~\ref{step:cap-lb} and Step~\ref{step:cap-net} we
can use the algorithm from the previous section. How Step~\ref{step:cap-approx} is
done will be explained later. First we show that the algorithm gives a
$(1+\eps,1+\delta)$-approximate solution. We start by showing that we get a valid
solution that violates the capacity constraint by at most a factor ~$1+\delta$.
%--------------------------------------------------------------------------------------------
\begin{lemma}\label{le:cap-correctness}
Let $\C^*:=\{b_1,\ldots,b_k\}$ be the set of squares computed in Step~\ref{step:cap-expand}.
There exists a partition $\{C_1,\ldots,C_k\}$ of $\Sq$ such that
$C_i\subseteq b_i$ and $|C_i|\leq (1+\delta) \cdot U_Q$ for each $1\leq i\leq k$,
and such a partition can be computed in $O(k^2 + n\log n)$ time.
\end{lemma}
%--------------------------------------------------------------------------------------------
\begin{proof}
Since $R$ is a weak $r$-packing, after expanding the squares in Step~\ref{step:cap-expand}
they cover all points in $\Sq$. Next we show that we can assign the points in $\Sq$
to the squares in $\C^*$ such that the capacities are not violated by
more than a factor~$1+\delta$.

Since $\C$ is a solution to the 0/1-weighted capacitated problem on $\Rq\cup\Aq$,
there is a partition $A_1,\ldots,A_k$ of $\Aq$ such that $A_i\subset b_i$
and $|A_i|\leq U$ for all $1\leq i\leq k$. Partition the plane into a
collection $\ranges$ of $O(k^2)$ cells by drawing the at most $2k$ vertical and $2k$ horizontal
lines containing the edges of the squares in~$\C$.
Consider a cell $\range\in\ranges$ and assume $\range$ is inside $j$~different
squares~$b_{i_1},\ldots,b_{i_j}\in\C$.
We can partition $\range$ into $j$ rectangular subcells $\range_1,\ldots,\range_{j}$
such that $|\Aq\cap\range_t| = |A_{i_t}\cap \range|$ for all $1\leq t\leq j$:
subcell~$\range_1$ will contain the topmost $|A_{i_1}\cap \range|$ points from ~$\Aq$,
subcell~$\range_2$ will contain the next $|A_{i_2}\cap \range|$ points,
and so on.
The total time for this is $O(n_{\range}\log n_{\range})$ time, where
$n_\range := |A_Q\cap\range|$.
We now assign all points from $\Sq \cap \range_{i_t}$ to the square $b_{i_t}$; in other words,
we put the points from $\Sq \cap \range_{i_t}$ into the cluster~$C_{i_t}$.
If we do this for all regions $\range\in\ranges$, we obtain the desired
partition $\{C_1,\ldots,C_k\}$ of $\Sq$.

It remains to prove that $|C_i|\leq (1+\delta) \cdot U_Q$ for each $1\leq i\leq k$.
Let $\ranges_i$ be the set of all subcells assigned to $C_i$.
Observe that $\sum_{\range\in\ranges_i} |\Aq\cap\range| = |A_i| \leq U$
and that\footnote{In fact, the description above would give $|\ranges_i|\leq 4k^2$.
However, in degenerate cases we may need two subcells for some $A_{i_t}$ when we
subdivide a cell $\sigma$, increasing the number of subcells in $\ranges_i$ by
at most a factor of~2.}
$|\ranges_i|\leq 8k^2$.
Moreover, since $A_Q$ is a
$\deltaq$-approximation for $\Sq$ we have
\[
|\Sq\cap\range|  \le \deltaq \cdot |\Sq| + |\Aq\cap\range|\cdot \frac{|\Sq|}{|\Aq|}.
\]

Hence,
\[
\begin{array}{lll}
|C_i| & = & \sum_{\range\in\ranges_i} |\Sq \cap \range| \\[2mm]
      & \leq & \sum_{\range\in\ranges_i} \big( \deltaq \cdot |\Sq| + |\Aq\cap\range|\cdot \frac{|\Sq|}{|\Aq|} \big) \\[2mm]
      & \leq &  |\ranges_i| \cdot \deltaq \cdot |\Sq| + \sum_{\range\in\ranges_i} |\Aq\cap\range|\cdot \frac{|\Sq|}{|\Aq|} \\[2mm]
      & \leq &  8k^2 \cdot \deltaq \cdot |\Sq| + U \cdot \frac{|\Sq|}{|\Aq|} \\[2mm]
      & \leq &  (\delta/(2k)) \cdot |\Sq| + \big((1+\delta/2)\cdot\alpha\cdot|\Aq|/k \big) \cdot \frac{|\Sq|}{|\Aq|} \\[2mm]
      & \leq &  (\delta/2) \cdot |\Sq|/k + (1+\delta/2)\cdot\alpha\cdot|\Sq|/k  \\[2mm]
      & \leq &  (1+\delta)\cdot U_Q \mbox{\hspace*{5mm} (since $\alpha\geq 1$)}
\end{array}
\]

To finish the proof, it remains to observe that the assignment of points to the
expanded squares described above can easily be done in $O(k^2+n\log n)$ time.
\end{proof}
%--------------------------------------------------------------------------------------------

We also need to prove that we get a $(1+\eps)$-approximate solution.
To this end, it suffices to show that an optimal solution $\Copt$ to the problem on~$\Sq$
is a valid solution on~$R\cup \Aq$. We can prove this by a similar approach
as in the proof of the previous lemma.
%--------------------------------------------------------------------------------------------
\begin{lemma}\label{le:cap-value}
The size of the squares in $\C^*$ is at most $(1+\eps)\cdot \Opt_k(\Sq,\alpha)$.
\end{lemma}
%--------------------------------------------------------------------------------------------
\begin{proof}
We show that an optimal solution $\Copt$ to the problem on~$\Sq$
is a valid solution on~$R\cup \Aq$. Let $\{b_1,\ldots,b_k\}$ and $\{C_1,\ldots,C_k\}$ be the
sets of squares and their corresponding clusters in a solution of value $\Opt_k(\Sq,\alpha)$
for $\Sq$. We claim that we can assign the points in $\Aq$ to the squares~$b_i$ such that
no square is assigned more than $U$ points, where $U:=(1+\delta/2)\alpha \cdot |\Aq|/k$.
We can do this following
a similar approach as in the proof of Lemma~\ref{le:cap-correctness}:
we partition the plane
into $O(k^2)$ cells, which we partition further into subcells that are assigned
to squares~$b_i$ such that $\sum_{\range\in\ranges_i} |\Sq\cap\range| = |C_i|$,
where $\ranges_i$ is the collection of subcells assigned to~$b_i$. Then
for each $b_i$ we have
\[
\begin{array}{lll}
\sum_{\range\in\ranges_i} |\Aq \cap \range|
      & \leq & \sum_{\range\in\ranges_i} \big( \deltaq \cdot |\Aq| + |\Sq\cap\range|\cdot \frac{|\Aq|}{|\Sq|} \big) \\[2mm]
      & \leq &  |\ranges_i| \cdot \deltaq \cdot |\Aq| + \sum_{\range\in\ranges_i} |\Sq\cap\range|\cdot \frac{|\Aq|}{|\Sq|} \\[2mm]
      & \leq &  8k^2 \cdot \deltaq \cdot |\Aq| + U_Q \cdot \frac{|\Aq|}{|\Sq|} \\[2mm]
      & \leq &  (\delta/(2k)) \cdot |\Aq| + \alpha \cdot |\Aq|/k \\[2mm]
      & \leq &  (\delta/2) \cdot \alpha\cdot|\Aq|/k + \alpha \cdot |\Aq|/k  \\[2mm]
      &  =   &  U \\
\end{array}
\]
\end{proof}
%--------------------------------------------------------------------------------------------

To make $\CapacitatedClusterQuery$ run efficiently, we need some more supporting data structures.
In particular, we need to quickly compute a $\deltaq$-approximation within our range~$\query$.
To this end, we use the following data structures.
\begin{itemize}
\item We compute a collection $A_1,\ldots,A_{\log n}$, where $A_i$ is a $(1/2^i)$-approximation on $\pointset$,
      using the algorithm of Phillips~\cite{p-aeat-08}. This algorithm computes, given a planar point set~$P$
      of size~$n$ and a parameter~$\delta$, a $\delta$-approximation of size
      $O((1/\delta) \log^4 (1/\delta)\cdot\polylog(\log(1/\delta)))$ in time
      $O((n/\delta^3)\cdot\polylog (1/\delta))$.
      We store each $A_i$ in a data structure for orthogonal range-reporting queries. If we use a range tree
      with fractional cascading, the data structure uses $O(|A_i|\log |A_i|)$ space
      and we can report all the points in $A_i \cap \query$ in time $O(\log n + |A_i \cap \query|)$.
\item We store $\pointset$ in a data structure for orthogonal range-counting queries. There is
      such a data structure that needs $O(n)$ space and it can answer orthogonal range-counting
      queries in $O(\log n)$ time~\cite{c-fads-88}.
\end{itemize} % kb: removed reference to ~\cite{bcko-cgaa-08}

We can now compute a $\deltaq$-approximation for $\Sq$ as follows.

%-----------------------------------------------------------------------------
\begin{algorithm}
\caption{$\DeltaSample(\query,\deltaq)$.}
\myenumerate{
\item \label{step:delta-resolution} Find the smallest value for $i$ such that
          $\frac{1}{2^i} \le \frac{\deltaq}{4}\frac{|\pointset_\query|}{|\pointset|}$, and
          compute $A := \query\cap A_i$.
\item \label{step:delta-sample} Compute a $(\deltaq/2)$-approximation $\Aq$ on
              $A$ using the algorithm by Phillips~\cite{p-aeat-08}.
\item \label{step:delta-slow} Return $\Aq$.
}
\end{algorithm}
%-----------------------------------------------------------------------------
%--------------------------------------------------------------------------------------------
%\begin{quotation}
%\noindent
%$\DeltaSample(\query,\deltaq)$ \\[-5mm]
%\begin{algorithmic}[1]
%  \State \label{step:delta-resolution} Find the smallest value for $i$ such that
%  $\frac{1}{2^i} \le \frac{\deltaq}{4}\frac{|\pointset_\query|}{|\pointset|}$, and
%  compute $A := \query\cap A_i$.
%  \State \label{step:delta-sample} Compute a $(\deltaq/2)$-approximation $\Aq$ on
%  $A$ using the algorithm by Phillips~\cite{p-aeat-08}.
%  \State \label{step:delta-slow} \Return $\Aq$.
%\end{algorithmic}
%\end{quotation}
%--------------------------------------------------------------------------------------------
\begin{lemma}\label{le:sample-correctness}
$\DeltaSample(\query,\deltaq)$ computes a $\deltaq$-approximation of size
\[
O((1/\deltaq)\cdot\polylog(1/\deltaq))
\]
on $\Sq$ in time
$O\left(\log^4(n/\deltaq)\cdot\polylog(\log n/\deltaq)\right)$.
%  \mm{I added the missing $\log n$ factor in polylog.}
\end{lemma}
%--------------------------------------------------------------------------------------------
To prove Lemma~\ref{le:sample-correctness} we need the following additional lemma.
%--------------------------------------------------------------------------------------------
\begin{lemma}
\label{lem:logsample}
If $A$ is a $\delta^*$-approximation for a point set $\pointset$ in $\Reals^2$ with
\[
\delta^* \le (\delta/2)\cdot (|\pointset_\query|/|\pointset|),
\]
then
$\sample_\query :=\query\cap\sample$ is a $\delta$-approximation for $\pointset_\query :=\pointset\cap\query$.
\end{lemma}
%--------------------------------------------------------------------------------------------
\begin{proof}
Consider any rectangular range $\range\subset\query$. Since $A$ is a $\delta^*$-approximation for $\pointset$
we have
\begin{equation*}
\left |\frac{|\pointset\cap\range|}{|\pointset|}-\frac{|\sample\cap\range|}{|\sample|}\right | \le \delta^*,
\end{equation*}
and so
\begin{equation}
\label{eq:range}
\left |\frac{|\pointset\cap\range|}{|\sample\cap\range|}-\frac{|\pointset|}{|\sample|}\right | \le \delta^* \frac{|\pointset|}{|\sample\cap\range|}.
\end{equation}

Similarly, by considering~$Q$ itself as a range we know that
\begin{equation*}
\left |\frac{|\Sq|}{|\pointset|}-\frac{|\sample_\query|}{|\sample|}\right | \le \delta^*
\end{equation*}
and so
\begin{equation}
\label{eq:query}
\left |\frac{|\pointset_\query|}{|\sample_\query|}-\frac{|\pointset|}{|\sample|}\right | \le \delta^* \frac{|\pointset|}{|\sample_\query|}.
\end{equation}

Combining Inequalities~(\ref{eq:range}) and~(\ref{eq:query}) and replacing $\delta^*$ with its upper bound we get
\begin{equation*}
\left |\frac{|\pointset_\query|}{|\sample_\query|}-\frac{|\pointset\cap\range|}{|\sample\cap\range|}\right |
\le
\frac{\delta}{2} \cdot \frac{|\pointset_\query|}{|\pointset|} \cdot \left (\frac{|\pointset|}{|\sample_\query|} + \frac{|\pointset|}{|\sample\cap\range|}\right )
=
\frac{\delta}{2} \cdot |\pointset_\query| \cdot \left (\frac{1}{|\sample_\query|} + \frac{1}{|\sample\cap\range|}\right).
\end{equation*}

Since $\sample\cap\range=\sample_\query\cap\range$ and $\pointset\cap\range=\pointset_\query\cap\range$,
and $|\sample_\query\cap\range|\le|\sample_\query|$, we can now derive
\begin{align*}
\label{eq:simplified}
\left |\frac{|\sample_\query\cap\range|}{|\sample_\query|}-\frac{|\pointset_\query\cap\range|}{|\pointset_\query|}\right |
& \le
  \frac{|\sample_\query\cap\range|}{|\pointset_\query|}
  \cdot
  \left |\frac{|\pointset_\query|}{|\sample_\query|}-\frac{|\pointset\cap\range|}{|\sample\cap\range|}\right |
\\[2mm]
& \le
  \frac{\delta}{2} \cdot \frac{|\sample_\query\cap\range|}{|\pointset_\query|}
  \cdot |\pointset_\query|
  \cdot \left( \frac{1}{|\sample_\query|} + \frac{1}{|\sample_\query\cap\range|} \right)
\\[2mm]
& \le \frac{\delta}{2}\left (\frac{|\sample_\query\cap\range|}{|\sample_\query|} + 1\right )
\\[2mm]
&\le \delta
\end{align*}
which proves the lemma.
\end{proof}
%-----------------------------------------------------------------------------------
Now we can prove Lemma~\ref{le:sample-correctness}.
%-----------------------------------------------------------------------------------
\begin{proof}
By Lemma~\ref{lem:logsample}, the set $A$ computed in Step~\ref{step:delta-resolution}
of \DeltaSample is a $(\deltaq/2)$-approximation for $\Sq$. Computing $A$ requires a
range query on $A_i$, which takes $O(\log n+|A|)$ time.
The $(1/2^i)$-approximation $A_i$ computed (during preprocessing) by Phillips's algorithm
has size
\[
|A_i| \ = \ O(2^i \cdot \log^4 (2^i) \cdot \polylog(\log 2^i))
      \ = \  O(2^i \cdot \log^4 n \cdot \polylog(\log n)).
\]

As $i$ is the smallest value with $1/2^i \le (\deltaq/4)\cdot (|\Sq|/|\pointset|)$,
we have $1/2^i > (\deltaq/8)\cdot (|\Sq|/|\pointset|)$. Hence,
\[
|\Sq|/|\pointset| < (8/\deltaq)\cdot (1/2^i)
\]

Since $A_i$ is a $(1/2^i)$-approximation for $\pointset$ we have
\[
 |A_i\cap\query|\leq (1/2^i) \cdot |A_i| + (|\Sq|/|\pointset|)\cdot |A_i|
\]
and so
\[
\begin{array}{lll}
|A| & = & |A_i \cap \query| \\[2mm]
    & \leq & (1/2^i) \cdot |A_i| + (|\Sq|/|\pointset|)\cdot |A_i| \\[2mm]
    & \leq & (1/2^i) \cdot |A_i| + (8/\deltaq)\cdot (1/2^i)\cdot |A_i| \\[2mm]
    & \leq & (1/2^i) \cdot |A_i|\cdot (1+ 8/\deltaq) \\[2mm]
    & =  & O(\log^4 n \cdot \polylog(\log n))\cdot O(1/\deltaq) \\[2mm]
    & =  & O((1/\deltaq)\cdot \log^4 n \cdot \polylog(\log n))\\
\end{array}
\]

Since a $\delta'$-approximation
of a $\delta''$-approximation of a set $P$ is a $(\delta'+\delta'')$-approximation
of $P$, we see that the set $\Aq$ computed in Step~\ref{step:delta-sample} is a
$\deltaq$-approximation, as required.
The time needed for Step~\ref{step:delta-sample} is
$O( (|A|/\deltaq^3)\cdot\polylog (1/\deltaq))$, which is
\[
O( (1/\deltaq)^4 \cdot \log^4 n \cdot \polylog (\log n/\deltaq)).
\]
\end{proof}

The only thing left is now an algorithm $\pckc(R\cup\Aq,k,U)$ that solves the 0/1-weighted
version of the capacitated rectilinear $k$-center problem. Here we use the following
straightforward approach. Let $m := |R\cup\Aq|$.
First we observe that at least one square in an optimal solution
has points on opposite edges. Hence, to find the optimal size we can do a binary
search over $O(m^2)$ values, namely the horizontal and vertical distances between
any pair of points. Moreover, given a target size~$s$ we can push al squares such
that each has a point on its bottom edge and a point on its left edge. Hence, to test
if there is a solution of a given target size~$s$, we only have to test $O(m^{2k})$
sets of $k$~squares. To test such a set~$\C = \{b_1,\ldots,b_k\}$ of squares,
we need to check if the squares cover all points in $R\cup \Aq$ and if we can assign
the points to squares such that the capacity constraint is met.
For the latter we need to solve a flow problem, which can be done in $O(m^2k)$ time.
More precisely, given a set $\C=\{b_1,\ldots,b_k\}$ of $k$ squares, a set $P$ of $m$ points,
and a capacity upper bound~$U$, and we have to decide if we can assign each point in $P$ to
a square in $\C$ containing it such that no square in $\C$ is assigned more than~$U$ points.
We can model this as a flow problem in a standard manner. For completeness we
describe how this is done.

We construct a flow network with source $s$ and sink $t$, and
one vertex $v_p$ for each point $p\in\Aq$ and one vertex $u_i$ for each square $b_i$.
We add the following edges.
\begin{enumerate}
\item For each $v_p$, we add one edge with capacity $1$ from $s$ to $v_p$.
\item For each $u_i$ we add one edge with capacity $|U|$ from $u_i$ to $t$.
\item For each pair $(p,b_i)$ where $p\in\Aq\cap b_i$ add an edge with capacity $1$ from $v_p$ to $u_i$.
\end{enumerate}

We solve the flow problem using the Ford-Fulkerson algorithm which works in
$O(|E| \cdot |f|)$ time, where
$|E|$ is the number of the edges and $|f|$ is maximum flow value.
In our problem, $|E|=O(mk)$ and $|f|=|U|\le m$, which results in an $O(m^2k)$ time bound.

Thus each step in the binary search takes $O(m^{2k+2} k)$, leading to
an overall time complexity for $\pckc(R\cup\Aq,k,U)$ of $O(m^{2k+2} k\log m)$,
where $m=|R\cup\Aq| = O\left( k/\eps^2 + (1/\deltaq)\cdot\polylog (1/\deltaq) \right)$,
where $\deltaq = \Theta(\delta/k^3)$.
\medskip

The following theorem summarizes the results in this section.
%--------------------------------------------------------------------------------------------
\begin{theorem}
Let $\pointset$ be a set of $n$ points in $\Reals^2$.
There is a data structure that uses $O(n\log n)$ storage
such that, for a query range~$\query$ and query values $k\geq 2$, $\eps>0$ and $\delta>0$,
we can compute a $(1+\eps,1+\delta)$-approximate answer to a rectilinear $k$-center query in
$O^*( (k/\eps)\log n  + ((k^3/\delta)\log n)^4 + ( k/\eps^2 + (k^3/\delta) )^{2k+2} )$
time, where the $O^*$-notation hides $O(\polylog (k/\delta))$ factors.
\end{theorem}
%--------------------------------------------------------------------------------------------

Note that for constant~$k$ and $\eps=\delta$ the query time simplifies to
$O^*((1/\eps^4)\log^4 n + (1/\eps)^{4k+4})$.
Also note that the time bound stated in the theorem only includes the time to compute the set
of squares defining the clustering. If we want to also report an appropriate assignment of points
to the squares, we have to add an $O(k^2+|\Sq|\log|\Sq|)$ term; see Lemma~\ref{le:cap-correctness}.
\medskip

%--------------------------------------------------------------------------------------------
\emph{Remark.}
The algorithm can be generalized to the rectilinear $k$-center problem in higher dimensions,
and to the Euclidean $k$-center problem; we only need to plug in an appropriate
appropriate $\delta$-approximation algorithm and an appropriate algorithm for the
0/1-weighted version of the problem.
%--------------------------------------------------------------------------------------------

%--------------------------------------------------------------------------------------------
\section{Exact \texorpdfstring{$k$}{k}-Center Queries in \texorpdfstring{$\Reals^1$}{1D}}\label{sec:1d}
%--------------------------------------------------------------------------------------------
In this section we consider $k$-center queries in~$\Reals^1$. Here we are given
a set $\pointset$ of $n$ points in $\Reals^1$
that we wish to preprocess into a data structure such that,
given a query interval $\query$ and a natural number~$k\geq 2$, we can
compute a set $\C$ of at most $k$ intervals of the same length that together cover all points in
$\Sq :=\pointset\cap\qrange$ and whose length is minimum.
We obtain the following result.
%--------------------------------------------------------------------------------------------
\begin{theorem}\label{thm:1d}
Let $\pointset$ be a set of $n$ points in $\Reals^1$.
There is a data structure that uses $O(n)$ storage
such that, for a query range~$\query$ and query value $k\geq 2$,
we can answer a rectilinear $k$-center query in $O(\min(k^2\log^2 n, 3^k\log n))$ time.
\end{theorem}
%-----------------------------------------------------------------------------------

The rest of the section is dedicated to the proof of the theorem.
Our data structure is simply a sorted array on the points in~$\pointset$ and therefore
it needs only $O(n)$ space, but it has two different query algorithms.
We call the query algorithms \emph{a query algorithm for large $k$} and
\emph{a query algorithm for small $k$}.
(See Section~\ref{subsec:largeK} and Section~\ref{subsec:smallK}.)
Both query algorithms start by shrinking the query interval $\query$ such that
its left and right endpoints coincide with a point in $\Sq$.
This can obviously be done in $O(\log n)$ time.
With a slight abuse of notation we still denote the shrunk interval by $\query$.
Let $x,x'$ be its left and right endpoints, respectively, so $\query=[x,x']$.

%-----------------------------------------------------------------------------------
\subsection{A Query Algorithm for Large~$k$}\label{subsec:largeK}
%-----------------------------------------------------------------------------------
This query algorithm uses a subroutine \Decider which, given an interval
$\query'$, a length~$\maxlen$ and integer $\ell\leq k$, can decide in $O(\ell\log n)$
time if all points in $\pointset\cap\query'$ can be covered by $\ell$
intervals of length~$\maxlen$. The global query algorithm then performs
a binary search, using \Decider as subroutine, to find a pair
of points $p_i,p_{i+1}\in\Sq$ such that the first interval in an optimal solution
covers $p_i$ but not $p_{i+1}$.
Then an optimal solution is found recursively for $k-1$ clusters within the query
interval $\query\cap [p_{i+1},\infty)$.
Next we describe the procedure \Decider.

\smallpara{The \Decider-procedure.}
The procedure \Decider takes as input an
integer $\ell$, a number $\maxlen$, and an interval $\query'=[a,a']$.
It returns \textsc{yes} if $\query'$ can be covered by at most $\ell$ subintervals of
length $\maxlen$, and \textsc{no} otherwise.
\Decider works as follows. Use binary search to find the first
point $p_i\in\pointset\cap\query'$ not covered by the interval $[a:a+\maxlen]$,
set $a := p_i$ and recurse. This continues until
either all points in $\pointset\cap\query'$ are covered, or more than $\ell$
intervals are used.
The \Decider runs in $O(\ell \cdot\log n)$ time and
outputs \textsc{yes} in the first case and outputs \textsc{no} in the
latter case.

\smallpara{The global query algorithm.}
Given $\qinterval$ and an integer $k$, we handle a query as follows.
Let $\Sq := \{p_i,\ldots,p_j\}$, where the points are numbered from left to right.
Thus $x=p_i$ and $x'=p_j$.
We do a binary search on $\{p_i,\ldots,p_j\}$ to find the smallest index
$i^*$ with $i\leq i^*\leq j$ such that $\Sq$ can be covered by $k$ intervals of
length $\maxlen := p_{i^*} - x$. Each decision in the binary search takes
$O(k\log n)$ time by a call to \Decider, so the entire binary search takes
$O(k\log^2 n)$ time.

Let $\Opt_k(P)$ denote the minimum interval length needed
to cover the points in a set $P$ by $k$ intervals. After finding $i^*$
we know that
\[
p_{i^*-1}-x < \Opt_k(\Sq) \leq p_{i^*}-x.
\]

If $\Opt_k(\Sq) < p_{i^*}-x$, then the first interval in an optimal solution
covers $\{p_i,\ldots,p_{i^*-1}\}$ and the remaining intervals cover~$\{p_{i^*},\ldots,p_{j}\}$.
Now we recursively compute $\Opt_{k-1}(\{p_{i^*},\ldots,p_{j}\})$, and since
\[
p_{i^*-1}-x <\Opt_{k-1}(\{p_{i^*},\ldots,p_{j}\})\leq p_{i^*}-x,
\]
we can safely report $\Opt_k(\Sq) = \Opt_{k-1}(\{p_{i^*},\ldots,p_{j}\})$.

It remains to analyze the running time of a query. The binary search takes
$O(k\log^2 n)$ times, after which we do a recursive call in which the
value of $k$ has decreased by~1. (The problem is easily
solved in $O(\log n)$ time when $k=1$.) Hence the number of recursive calls is $k$, leading
to an $O(k^2 \log^2 n)$ query time, as claimed.  Finding an optimal solution---and not
just the value of an optimal solution---can be done within the same time bound.
We get the following lemma.
%-----------------------------------------------------------------------------------
\begin{lemma}\label{lem:1d-version1}
Let $\pointset$ be a set of $n$ points in $\Reals^1$.
There is a data structure that uses $O(n)$ storage
such that, for a query range~$\query$ and query value $k\geq 2$,
we can answer a rectilinear $k$-center query in $O(k^2\log^2 n)$ time.
\end{lemma}
%-----------------------------------------------------------------------------------

%-----------------------------------------------------------------------------------
\subsection{A Query Algorithm for Small~$k$}\label{subsec:smallK}
%-----------------------------------------------------------------------------------
Here we present the second query algorithm of the data structure, which is more
efficient for small values of~$k$. We begin with the following definition.

\begin{definition}\label{def:fair-split}
    Let $\Sq$ be a set of points inside a query interval~$\query=[x,x']$,
    such that $x,x'\in\Sq$. We call a point $r\in\query$ a \emph{fair split point} if
    there is an optimal solution $\clusteringopt(\query) := \{I_1,I_2,\ldots,I_k\}$
    for the $k$-center problem on $\Sq$  such that
     \begin{enumerate}
     \item[(i)] $r$ does not lie in the interior of any interval $I_j\in \clusteringopt(\query)$,
     and
     \item[(ii)] the number of intervals in $\clusteringopt(\query)$ lying to the left of $r$ is $k(r-x)/(x'-x)$.
     % and the number of intervals in $\clusteringopt_k(\query)}$ lying to the right of $r$
     % is $k\frac{x'-r}{x'-x}$.
     % \mdb{second part is implied by first part
     \end{enumerate}
\end{definition}	

Note that the split point $r$ is not necessarily a point in $\Sq$, that is, it is not one of
the given points. The following lemma is crucial in our analysis.
%-----------------------------------------------------------------------------------
\begin{lemma}\label{lem:fair-split}
Let $\splitset:=\{s_1,s_2,\ldots,s_{k-1}\}$ denote the set of points that partition $\query$ into
$k$ equal-size subintervals. Then at least one of the points of $\splitset$ is a fair split point.
\end{lemma}
%-----------------------------------------------------------------------------------
\begin{proof}
	First we prove that there exists a point in $\splitset$ that does not lie in
    the interior of some $I_j\in \clusteringopt(\query)$.
    To this end, we observe that if the length of optimal intervals equals $(x'-x)/k$,
    then the optimal solution
    is equal to the subdivision of $\query$ defined by the split points, and so the
    lemma trivially holds.
    Otherwise, the length of optimal intervals is strictly smaller than $(x'-x)/k$.
    But then an interval in
    $\clusteringopt(\query)$ can contain at most one point from $\{s_0,\ldots s_k\}$, where
    $s_0:= x$ and $s_k:=x'$.
    Since $s_0$ and $s_k$ are points in $\Sq$, there is an interval in $\clusteringopt(\query)$
    containing $s_0$ and one containing $s_k$. Hence, the remaining $k-2$ intervals in
    $\clusteringopt(\query)$ can cover at most $k-2$ points from the split points
    $\{s_1,\ldots s_{k-1}\}$ and so at least one of the split points will not be
    covered by the union of the subintervals.

    It remains to prove that for at least one of the points of $\splitset$ that satisfies
    Condition~(i) in Definition~\ref{def:fair-split}, it also satisfies Condition~(ii)
    in Definition~\ref{def:fair-split}.
    First consider the case that there is only one $s_i$ with $0 < i < k$ that is not
    the interior of any $I_j$. Let $\ell_i := |\{ I_j : I_j \subset [s_0, s_i]\}|$ denote
    the number of intervals to the left of $s_i$, and let and $f_i := \ell_i-i$.
    Since all the $s_j$ with $s_0 \leq s_j<s_i$ are contained in distinct intervals from
    $\clusteringopt(\query)$, we have $f_i \geq 0$. But since the same holds for all $s_j$
    with $s_i < s_j \leq s_k$, the number of intervals to the right of $s_i$ is at least~$k-i$.
    Hence, $f_i \leq k - (k-i) -i = 0$. We conclude that $f_i = 0$, so $s_i$ is a fair split
    point.

    Next we consider the case that several $s_i$ are not in the interior of any $I_j$.
    Let $0 < i_1 < \ldots < i_m < k$ be the corresponding indices. By the same arguments as above
    we have $f_{i_1} \geq 0$ and $f_{i_m} \leq 0$. Furthermore the sequence $\ell_i$ is non-decreasing,
    which implies $f_{i_j + 1} \geq f_{i_j}-1$. As a consequence, there is an $i_j$ with $f_{i_j} = 0$.
    It follows that $s_{i_j}$ is a fair split point.
\end{proof}
%-----------------------------------------------------------------------------------

Lemma~\ref{lem:fair-split} suggests the following approach.
Again, the data structure is just a sorted array on the points in $\pointset$.
  A query with range $\query=[x,x']$ and parameter $k$ is answered as follows.
  Search the array for the successor $s(x)$ of $x$ and the predecessor
  $p(x')$ of $x'$ in $\pointset$. Replace $\query$ with $[s(x),p(x')]$, so that
  the left and right endpoints of the modified range $\query$ are points from~$\pointset$.
  Partition $\query$ into $k$ equal-size subintervals.
  At each split point $s_i$ of $\query$, recursively solve the problem on
  $\query_{\text{left}} := [x,s_i]$ with parameter $k_\text{left} := i$ and on
  $\query_{\text{left}} := [s_i,x']$ with parameter $k_\text{right} := k-i$.
  By Lemma~\ref{lem:fair-split}, at (at least) one of the split points of $\query$
  the union of the returned intervals is an optimal solution. Moreover, we can easily
  maintain the best solution as we try all split points, so that after trying
  all split points we can return an optimal solution.

  The recursion ends when $k=1$. In this case we report $[s(x),p(x')]$ as the optimal solution.
We obtain the following result.
%-----------------------------------------------------------------------------------
\begin{lemma}\label{lem:1d-version2}
Let $\pointset$ be a set of $n$ points in $\Reals^1$.
There is a data structure that uses $O(n)$ storage
such that, for a query range~$\query$ and query value $k\geq 2$,
we can answer a rectilinear $k$-center query in $O(3^k\log n)$ time.
\end{lemma}
%-----------------------------------------------------------------------------------
\begin{proof}
  It takes $\runtime(\log n)$ time to find the successor and the predecessor of $x$
  and $x'$ in $\pointset$.
  % If $k=1$ then the optimal solution is $[s(x),p(x')]$ itself and we are done.
  % If $k=2$ then the optimal
  % cover consists of the two intervals that meet each other at the middle point $m(q)$
  % of $q$ if $m(q)$ is a point of $\pointset$. If $m(q)$ is not a point of $\pointset$
  % then the optimal cover consists of the following two intervals:
  % ($i$) the interval that starts at $s(x)$ and ends at the predecessor of $m(q)$, and
  % ($ii$) the interval that starts at the successor of $m(q)$ and ends at $p(x')$.
  Hence, we obtain the following recurrence for the time $T(k,n)$ needed to answer
  a $k$-center query on a point set of size~$n$:
  \[
  T(k,n) \ \leq \ \left\{ \begin{array}{ll}
                    O(\log n) & \mbox{if $k=1$} \\
                    O(\log n) + \sum_{i=1}^{k-1} T(i,n)+T(k-i,n) & \mbox{if $k>1$}
                    \end{array}
             \right.
\]
which solves to $T(n,k)=O(3^k \log n)$. To see this, note that the
for recurrence
\[
T^*(k) = \sum_{i=1}^{k-1} T^*(i)+T^*(k-i)
\]
we have
\[
T^*(k) = 2\sum_{i=1}^{k-1} T^*(i) = 3 T^*(k-1),
\]
so with $T^*(1)=1$ we obtain $T^*(k) = 3^{k-1}$, which implies $T(n,k)=O(3^k \log n)$.
\end{proof}

%--------------------------------------------------------------------------------------------
\section{Exact Rectilinear 2- and 3-Center Queries in \texorpdfstring{$\Reals^2$}{the Plane}}\label{sec:2d}
%--------------------------------------------------------------------------------------------
Suppose we are given a set $\pointset=\{\point_1,\point_2,\ldots,\point_n\}$ of $\setsize$
points in $\Reals^2$ and an integer $k$. In this section we build a data structure $\ds$ that
stores the set $\pointset$ and, given an orthogonal query rectangle~$\query$, can be used to quickly
find an optimal solution for the $k$-center problem on $\Sq:= \pointset\cap\query$ for $k=2$ or~3.

\subsection{2-Center Queries}\label{sec:k2}
We begin by a quick overview of our approach.
%\textbf{Idea. }\label{par:2d-k2-idea}
%Our approach for the case of $k=2$ is as follows.
We start by shrinking the query range $\query$ such that each edge
of $\query$ touches at least one point
of $\pointset$. (The time for this step is subsumed by the time for the rest of the procedure.)
It is well known that if we want to cover $\Sq$
by two squares $\sigma,\sigma'$ of minimum size, then $\sigma$ and $\sigma'$ both
share a corner with $\query$
and these corners are opposite corners of $\query$. We say that $\sigma$ and $\sigma'$
are \emph{anchored} at the corner they share with~$\query$.
Thus we need to find optimal solutions for the two cases---$\sigma$ and $\sigma'$ are
anchored at the topleft and bottomright corner of $\query$, or at the  topright
and bottomleft corner---and return the better one.
Let $\corner$ and $\corner'$ be the topleft and the bottomright corners of~$\query$.
In the following we describe how to compute two squares $\sigma$ and $\sigma'$ of minimum size
that are anchored at $\corner$ and $\corner'$, respectively, and whose union covers~$\Sq$.
The topright/bottomleft case can then be handled in the same way.

\begin{figure}
  \centering
  \includegraphics[width=13cm]{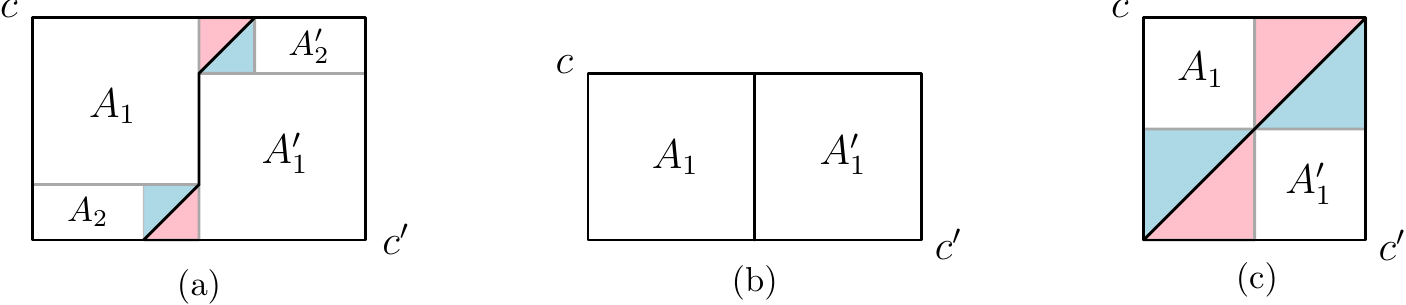}
    \caption{Various types of $\infmetric$-bisectors. The bisectors are shown in blue.
  (a): $\qrange$ is \q{fat}. The regions $A_j,A'_j$ for $j=1,2$ are shown with text.
  (b): $\qrange$ is \q{thin}.
  The regions $A_j$ and $A'_j$ for $j=2,3,4$ are empty.
  (c): $\qrange$ is a square. The regions $A_j$ and $A'_j$ for $j=2$ are empty.
  In both (a) and (c) regions $A_3,A'_3$ are colored in blue and $A_4,A'_4$ are colored
  in red.}
  \label{fig:L-infty-bisector}
\end{figure}

First we determine the $\infmetric$-bisector of $\corner$ and
$\corner'$ inside $\qrange$; see Figure~\ref{fig:L-infty-bisector}.
The bisector partitions $\qrange$ into regions $\region$ and $\region'$, that respectively have
$\corner$ and $\corner'$ on their boundary.
Obviously in an optimal solution (of the type we are focusing on),
the square $\sigma$ must cover $\Sq\cap \region$ and
the square $\sigma'$ must cover $\Sq\cap \region'$.
To compute $\sigma$ and $\sigma'$, we thus need to find the points $q\in\region$ and $q'\in\region'$
with maximum $\infmetric$-distance to the corners $\corner$ and $\corner'$, respectively.
To this end, we partition $\region$ and $\region'$ into subregions such that in each of the
subregions the point with maximum $\infmetric$-distance to its corresponding corner can be
found quickly via appropriate data structures discussed below. We assume w.l.o.g.
that the $x$-span of $\qrange$ is at least its $y$-span.
We begin by presenting the details of such a partitioning for Case~(a) of
Figure~\ref{fig:L-infty-bisector}---Case~(b) and Case~(c) can be seen as special cases
of Case~(a).

As Figure~\ref{fig:L-infty-bisector} suggests, we partition $\region$ and
$\region'$ into subregions. We denote these subregions by $A_j$ and $A'_j$,
for $1\leq j\leq 4$. From now on we focus on reporting the
point~$q\in\pointset$ in~$\region$ with maximum $\infmetric$-distance to $\corner$;
finding the furthest point from $\corner'$ inside~$\region'$ can be done similarly.
Define four points $\point(A_j)\in\pointset$ for $1\leq j\leq 4$ as follows.
\begin{itemize}
\item The point $\point(A_1)$ is the point of $\Sq$
      with maximum $\infmetric$-distance to $\corner$ in $A_1$.
      Note that this is either the point with maximum $x$-coordinate
      in $A_1$ or the point with minimum $y$-coordinate.
\item The point $\point(A_2)$ is a bottommost point in $A_2$.
\item The point $\point(A_3)$ is a bottommost point in $A_3$.
\item The point $\point(A_4)$ is a rightmost point in $A_4$.
\end{itemize}

Clearly
\begin{equation}
 q = \arg\max_{1\leq j\leq 4} \ \{\infdist(\point(A_j),\corner)\},
\end{equation}
where $\infdist(.)$ denotes the $\infmetric$-distance function.

\smallpara{Data structure.}   % \label{par:2d-k2-ds}
Our data structure now consists of the following components.
\begin{itemize}
\item We store $\pointset$ in a data structure $\ds_1$ that allows us to
    report the extreme points in the $x$-direction and in the $y$-direction inside a
    rectangular query range. For this we use the structure by Chazelle~\cite{c-fads-88},
    which needs $O(n\log^{\delta} n)$ space and has $O(\log n)$ query time, where
    $\delta>0$ is an arbitrary small (but fixed) constant.
\item We store $\pointset$ in a data structure $\ds_2$ with two components.
    The first component should answer the following queries: given a
    $45^{\circ}$ query cone whose top bounding line is horizontal
    and that is directed to the left---we obtain such a cone when we
    extend the region $A_4$ into an infinite cone---, report the
    rightmost point inside the cone.
    The second component should answer similar queries for cones that are the
    extension of~$A_3$.
\end{itemize}

Lemma~\ref{lem:2d-k2-45cone} proves the existence of a linear-size data structure
that implements such a component and that has $O(\log n)$ query time.
\begin{lemma}\label{lem:2d-k2-45cone}
Each component of $\ds_2$ has complexity $\runtime(n)$ and it can be built in $\runtime(n \log n)$
time.
\end{lemma}
\begin{proof}
We describe the component for the following queries: given a
$45^{\circ}$ query cone whose bottom bounding line is horizontal
and that is directed to the right, report the leftmost point inside the cone.
Our structure for such queries is defined as follows. For each point
$p_i\in\point$ consider the inverted cone with apex~$p_i$, that is, the
$45^{\circ}$ cone whose top edge is horizontal. We now add these
inverted cones from right to left, where we add each cone ``on top of'' the
existing cones. This gives us a linear-size subdivision, which
is a Voronoi diagram for the distance function induced by our problem, which we
preprocess for point location. If we then do a point-location query
in the subdivision with the apex of our query cone, then this tells
us the leftmost point inside the query cone.

To construct the structure, we use a sweep-line approach.
The sweep line is a vertical line that moves from right to left. The sweep line halts at each
point $\point_i\in\pointset$, and computes the Voronoi cell of $\point_i$, denoted with
\vor($\point_i$), as the set of all the points in the plane that lie in the unbounded left
$45^\circ$-cone starting at $\point_i$.
If \vor($\point_i$) intersects \vor($\point_j$), for some $j<i$, then the region
\vor($\point_i$)$\subseteq$\vor($\point_j$) will belong to \vor($\point_i$).
Observe that \vor($\point_i$) can intersect at most one \vor($\point_j$) with $j<i$ and therefore
updating \vor($\point_i$) can be done easily.
See Figure~\ref{fig:45-voronoi-diagram} for a picture of execution of the algorithm for a
few successive iterations.
%%%%%%%%%%%%%%%%%%%%%%%%%%%%%%%%%%%%%%%%%%%%%%%%%%%%%%%%%%%%%%%%%%%%%%%%%%%%%%%%%%%%%%%%%%%
\begin{figure}
  \centering
  \includegraphics[width=12cm]{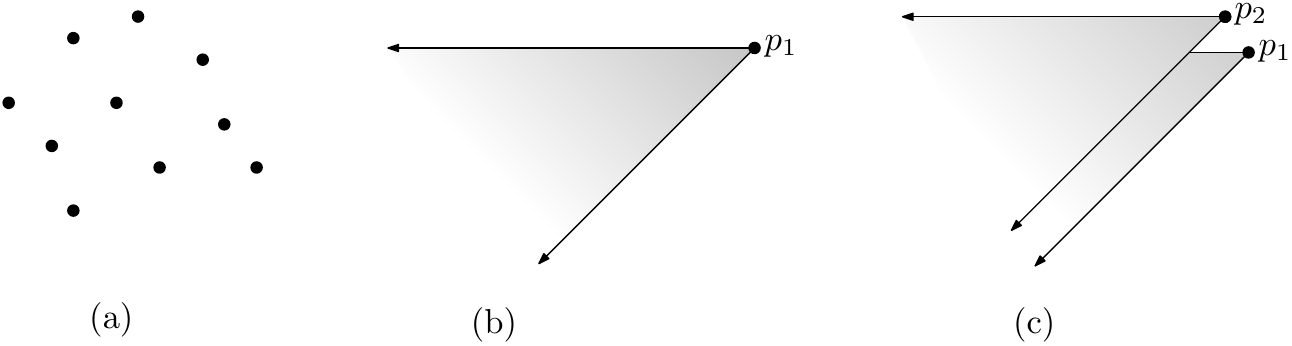}
  \caption{A point set $\pointset$ and the Voronoi cells of the first two points of $\pointset$
           visited by the sweep-line algorithm described in Lemma~\ref{lem:2d-k2-45cone}.}
  \label{fig:45-voronoi-diagram}
\end{figure}
%%%%%%%%%%%%%%%%%%%%%%%%%%%%%%%%%%%%%%%%%%%%%%%%%%%%%%%%%%%%%%%%%%%%%%%%%%%%%%%%%%%%%%%%%%%
\end{proof}

\smallpara{Query procedure. } %\label{par:2d-k2-query}
Given an axis-aligned query rectangle $\qrange$, we first (as already mentioned) shrink the
query range so that each edge of $\qrange$ contains at least on point of $\pointset$.
Then compute the $\infmetric$-bisector of $\qrange$. Query $\ds_1$ with $A_1$ and
$A_2$, respectively, to get the points $\point(A_1)$ and $\point(A_2)$.
Then query $\ds_2$ with $\qpoint$ and $\qpoint'$ to get the points $\point(A_3)$ and
$\point(A_4)$, where $\qpoint$ and $\qpoint'$ are respectively the bottom and the top
intersection points of $\infmetric$-bisector of $\qrange$ and the boundary of $\qrange$.
Among the at most four reported points, take the one with maximum $\infmetric$-distance the corner
$\corner$. This is the point $q\in\Sq\cap\region$ furthest from~$\corner$.

Compute the point $q'\in\Sq\cap\region$ furthest from~$\corner'$ in a similar fashion.
Finally, report two minimum-size congruent squares $\sigma$ and $\sigma'$ anchored
at $\corner$ and $\corner'$ and containing $q$ and $q'$, respectively.
\medskip

Putting everything together, we end up with the following theorem.
\medskip

\begin{theorem}\label{thm:2d-k2}
Let $\pointset$ be a set of $n$ points in the plane. For any fixed $\delta>0$,
there is a data structure using $O(n\log^{\delta} n)$ space
that can answer rectilinear $2$-center queries in $O(\log n)$ time.
\end{theorem}

\emph{Remark.}
We note that the query time in Theorem~\ref{thm:2d-k2} can be improved
in the word-RAM model to $\runtime(\log\log n)$ by using the range successor data
structure of Zhou~\cite{zhou-ipl-16}, and the point-location data structure for
orthogonal subdivisions by de Berg~\etal~\cite{bks-jal-95}.

\subsection{3-Center Queries}\label{sec:k3}
Given a (shrunk) query range $\qrange$, we need to compute a set
$\{\sigma,\sigma',\sigma''\}$ of (at most) three congruent squares of minimal size
whose union covers $\Sq$. It is easy to verify (and is well-known) that at least one of
the squares in an optimal solution must be anchored at one of the corners of
$\qrange$.
Hence and w.l.o.g. we assume that $\sigma$ is anchored at one of the corners of $\qrange$.
We try placing $\sigma$ in each corner of $\qrange$ and select the placement
resulting in the best overall solution.
Next we briefly explain how to find the best solution subject to placing $\sigma$
in the leftbottom corner of $\qrange$. The other cases are symmetric.
We perform two separate binary searches; one will test placements of $\sigma$ such that
its right side has the same $x$-coordinate as a point in $\pointset$, the other will be on
possible $y$-coordinates for the top side.
During each of the binary searches, we compute the smallest axis-parallel rectangle
$\qrange'\subseteq\qrange$ containing the points of $\qrange\backslash\sigma$ (by covering
$\qrange\backslash\sigma$ with axis-aligned rectangles and querying for extreme points in
these rectangles).
We then run the algorithm for $k=2$ on $\qrange'$. We need to ensure that this query ignores
the points already covered by $\sigma$. For this, recall that for $k=2$ we covered the regions
$A$ and $A'$ by suitable rectangular and triangular ranges. We can now do the same, but we cover
$A\setminus\sigma$ and $A'\setminus\sigma$ instead.

After the query on $\qrange'$, we compare the size of the resulting squares with the size of
$\sigma$ to guide the binary search.
%If the largest sizes of $\sigma_1$
%and the algorithm for $k=2$ belongs to $\sigma_1$, we guide the binary search to the
%left (resp. down) and repeat the process.
 The process stops as soon as the three sizes
are the same or no further progress in the binary search can be made.
\medskip

Putting everything together, we end up with the following theorem.
%-----------------------------------------------------------------------------------
\begin{theorem}\label{thm:2d-k3}
Let $\pointset$ be a set of $n$ points in the plane. For any fixed $\delta>0$,
there is a data structure using $O(n\log^{\delta} n)$ space
that can answer rectilinear $3$-center queries in $O(\log^2 n)$ time.
\end{theorem}
%-----------------------------------------------------------------------------------
\emph{Remark.}
Similar to Theorem~\ref{thm:2d-k2}, the query time in Theorem~\ref{thm:2d-k3} can
be improved in the word-RAM model of computation to $\runtime(\log n\log\log n)$ time.

\section{Discussion}\label{sec:ch4-discussion}
In this paper we presented a general method to preprocess a given point
set $\pointset$ in $\Reals^d$ into a data structure for fast range-clustering
queries on the subset of $\pointset$ that lies inside a given axis-aligned query
box $\qrange$. Our main result is a general method to compute a $(1+\eps)$-approximation
to a range-clustering query, where $\eps>0$ is a parameter that can be specified
as part of the query.
Our method applies to a large class of clustering problems, including $k$-center
clustering in any $\Lp$-metric and a variant of $k$-center clustering where the
goal is to minimize the sum (instead of maximum) of the cluster sizes.
We also extended our method to deal with capacitated $k$-clustering problems, where
each cluster should contain at most a given number of points.
For the special cases of rectilinear $k$-center clustering in $\Reals^1$ and in
$\Reals^2$ for $k=2$ or~3, we described data structures that answer range-clustering
queries exactly.
\medskip

We close the paper by stating the following open questions.
\begin{itemize}
  \item Can the bound in Theorem~\ref{thm:noncapacitated} (and the bounds in Corollary~\ref{cor:noncapacitated})
        be improved?
  \item Is it possible to design efficient exact data structures for rectilinear $k$-center queries
        when $k>3$?
  \item Can any of the data structures presented in this paper be made dynamic?
  \item Is it possible to extend our results on approximate queries to non-regular cost functions
  (for example, for the $k$-means problem)?
\end{itemize}
\medskip

\textbf{Acknowledgements. }
This research was initiated when the first author visited
the Department of Computer Science
at TU Eindhoven during the winter 2015--2016.
He wishes to express his
gratitude to the other authors and the department for their
hospitality.
The last author wishes to thank Timothy Chan for valuable
discussions about the problems studied in this paper.

%\bibliography{references}
\newcommand{\SortNoop}[1]{}

\end{document}